\documentclass{article}

\let\articleaddcontentsline\addcontentsline
\usepackage{iclr2027_conference,times}
\let\addcontentsline\articleaddcontentsline

\usepackage[T1]{fontenc}
\usepackage[utf8]{inputenc}
\usepackage{microtype}
\usepackage{amsmath,amssymb,amsfonts,amsthm,mathtools}
\usepackage{aliascnt}
\usepackage{bm}
\allowdisplaybreaks

\usepackage{graphicx}
\usepackage{booktabs,array}
\usepackage{pifont}
\usepackage{enumitem}
\usepackage{needspace}
\usepackage{placeins,flafter}
\usepackage[ruled,vlined]{algorithm2e}
\SetAlFnt{\small}
\SetAlCapFnt{\small}
\SetAlCapNameFnt{\small}
\SetAlCapHSkip{0pt}
\IncMargin{-\parindent}

\usepackage{xcolor}
\usepackage{hyperref}
\usepackage{url}
\usepackage[capitalize,noabbrev,nameinlink]{cleveref}
\hypersetup{hidelinks,pdftitle={A Logarithmic Regret Bound for Optimistic Hedge in General-Sum Games},pdfauthor={Anonymous}}
\newcommand{\cmark}{\textcolor{green!80!black}{\ding{51}}}
\newcommand{\xmark}{\textcolor{red}{\ding{55}}}

\theoremstyle{plain}

\newtheoremstyle{centralquestionstyle}{\topsep}{\topsep}
  {\itshape\centering}{}{}{}{0pt}{}
\theoremstyle{centralquestionstyle}
\newtheorem*{centralquestion}{}
\theoremstyle{plain}

\newaliascnt{lemma}{theorem}
\newtheorem{lemma}[lemma]{Lemma}
\aliascntresetthe{lemma}

\newaliascnt{proposition}{theorem}

\aliascntresetthe{proposition}

\newaliascnt{corollary}{theorem}

\aliascntresetthe{corollary}

\theoremstyle{plain}
\newaliascnt{definition}{theorem}
\newtheorem{definition}[definition]{Definition}
\aliascntresetthe{definition}

\theoremstyle{definition}
\newaliascnt{assumption}{theorem}

\aliascntresetthe{assumption}

\newaliascnt{example}{theorem}

\aliascntresetthe{example}

\theoremstyle{remark}
\newaliascnt{remark}{theorem}

\aliascntresetthe{remark}

\theoremstyle{plain}
\usepackage{thm-restate}
\newcommand{\printrepeatedstatement}[2]{%
  \begin{NoHyper}#2*\end{NoHyper}%
}

\crefname{theorem}{Theorem}{Theorems}
\Crefname{theorem}{Theorem}{Theorems}
\crefname{lemma}{Lemma}{Lemmas}
\Crefname{lemma}{Lemma}{Lemmas}
\crefname{proposition}{Proposition}{Propositions}
\Crefname{proposition}{Proposition}{Propositions}
\crefname{corollary}{Corollary}{Corollaries}
\Crefname{corollary}{Corollary}{Corollaries}
\crefname{remark}{Remark}{Remarks}
\Crefname{remark}{Remark}{Remarks}
\crefname{definition}{Definition}{Definitions}
\Crefname{definition}{Definition}{Definitions}
\crefname{equation}{Equation}{Equations}
\Crefname{equation}{Equation}{Equations}
\crefname{section}{Section}{Sections}
\Crefname{section}{Section}{Sections}

\newcommand{\R}{\mathbb{R}}
\newcommand{\E}{\mathbb{E}}
\newcommand{\1}{\mathbf{1}}
\newcommand{\Reg}{\operatorname{Reg}}
\newcommand{\Var}{\operatorname{Var}}
\newcommand{\softmax}{\operatorname{softmax}}
\newcommand{\osc}{\operatorname{osc}}
\newcommand{\diag}{\operatorname{Diag}}
\newcommand{\ip}[2]{\left\langle #1,#2\right\rangle}
\newcommand{\norm}[1]{\left\lVert #1\right\rVert}
\newcommand{\abs}[1]{\left|#1\right|}
\newcommand{\dd}{\mathrm{d}}
\newcommand{\appref}[1]{\hyperref[#1]{Appendix~\ref*{#1}}}

\title{A Logarithmic Regret Bound for Optimistic Hedge in General-Sum Games}

\author{Junsoo Ha \\
\texttt{junsoo.ha.contact@gmail.com}}

\iclrfinalcopy  
\ificlrfinal\hypersetup{pdfauthor={Junsoo Ha}}\fi

\begin{document}

\maketitle
\lhead{Preprint}

\begin{abstract}
Can simple no-regret dynamics attain smaller regret in self-play than against worst adversaries? In $n$-player general-sum games with $d_i$ actions for each player $i$, \citet{DFG2021} proved an $O(n\log d_i\log^4 T)$ individual regret bound for Optimistic Hedge, which improves upon classical $O(\sqrt T)$ adversarial regret bound. We show that Optimistic Hedge with a constant step size achieves $O(\sqrt n\log d_i\log T)$ individual external regret under expected loss-vector feedback. The time-averaged play consequently enjoys a coarse correlated equilibrium gap $O(\sqrt n\log d\log T/T)$, where $d=\max_i d_i$. The improvement comes from a larger admissible step size $\eta=\Theta(1/(\sqrt n\log T))$. Our analysis proves factorial bounds on high-order differences of probability-weighted pairwise loss gaps, then applies finite-difference interpolation in a fixed Euclidean norm. These estimates sharpen the analysis of \citet{DFG2021} and yield a logarithmic regret bound.
\end{abstract}

\section{Introduction}
\label{sec:introduction}

Game theory studies strategic interaction \citep{vonNeumannMorgenstern1944,Nash1951}. How can independent players learn to approach equilibrium \citep{FudenbergLevine1998}? The same question arises in poker self-play \citep{BrownSandholm2019Pluribus}, adversarial training \citep{GoodfellowEtAl2014GAN}, and language-model alignment \citep{MunosEtAl2024NLHF,WuEtAl2026MNPO}.

No-regret learning is a decentralized route. Players compete with their best fixed action in hindsight \citep{Blackwell1956,Hannan1957}. Sublinear individual regret drives average play to coarse correlated equilibrium (CCE) \citep{HartMasColell2001,CBL2006}. While adversarial minimax regret is tightly bounded as $\Theta(\sqrt T)$ \citep{FreundSchapire1997,CesaBianchiEtAl1997}, self-play offers an additional structure: each player's losses evolve with the same update rule.

Optimistic methods exploit such structure and predict future losses from the past \citep{RakhlinSridharan2013,SALS2015}. For Optimistic Hedge, \citet{DFG2021} proved $O(\log^4 T)$ individual regret, and subsequent work achieved logarithmic or constant regret by modifying the dynamics \citep{FarinaEtAl2022,SoleymaniEtAl2025,SoleymaniEtAl2025Fast,AbbadiLarakiMertikopoulos2026}. This line of recent progress leaves us with the following natural question:

\begin{centralquestion}
Can plain Optimistic Hedge achieve logarithmic individual regret in general-sum games?
\end{centralquestion}

\paragraph{Contributions.}
We answer the question affirmatively. Our contributions are two-fold.
\begin{enumerate}[label=(\roman*),leftmargin=*,itemsep=3pt,topsep=2pt]
    \brokenpenalty=10000
    \item \textbf{An improved high-order smoothness analysis.} Building on the high-order smoothness framework of \citet{DFG2021}, we construct a centered-logit recurrence map that transfers derivative bounds to sharper factorial bounds on high-order weighted-gap differences. Then, a direct finite-difference interpolation produces a larger admissible step size. Our smoothness analysis based on a recurrence map may be of independent interest for regret analysis in games.
    \item \textbf{Improved regret bound.} For $n$ players with action dimension $d_i$, we obtain $O(\sqrt n\log d_i\log T)$ individual regret; this improves upon the $O(n\log d_i\log^4 T)$ bound of \citet{DFG2021}. Our result suggests that Optimistic Hedge with a constant step size already suffices: neither modified dynamics nor an adaptive learning rate is necessary.
\end{enumerate}

\Cref{sec:related,sec:setup} review prior work and the setup; \cref{sec:main-results,sec:analysis} present the results and analysis.

\section{Related work}
\label{sec:related}

\paragraph{No-regret learning and games.}
No-regret learning connects individual performance guarantees to equilibrium behavior. Its classical foundations compare a learner's cumulative loss with that of the best fixed action in hindsight \citep{Blackwell1956,Hannan1957}, and algorithms such as Weighted Majority and Hedge achieve this guarantee against arbitrary loss sequences \citep{LittlestoneWarmuth1994,FreundSchapire1997}. When all players learn in this way, their individual guarantees combine: averaging yields approximate Nash equilibria in two-player zero-sum games and coarse correlated equilibria in general-sum games \citep{FreundSchapire1999,MoulinVial1978,HartMasColell2001}. Controlling deviations that depend on a mediator's recommendation gives the stronger notion of correlated equilibrium, motivating connections with calibration, regret matching, and swap regret \citep{Aumann1974,FosterVohra1997,HartMasColell2000,BlumMansour2007}. These connections make regret analysis a key quantitative tool for equilibrium computation: better regret bounds translate directly into faster equilibrium approximation.

\paragraph{Accelerated rates under self-play.}
Self-play offers a crucial structure that admits improved rates: each player's losses evolve with the same shared update rule. Optimistic methods exploit this structure by using past losses to predict future ones \citep{RakhlinSridharan2013,RakhlinSridharan2013Predictable}. Early logarithmic-regret guarantees in two-player zero-sum games demonstrated that faster learning is indeed possible under self-play \citep{DaskalakisDeckelbaumKim2011}, and subsequent work developed increasingly sharp guarantees for optimistic dynamics in general-sum games \citep{SALS2015,ChenPeng2020}. One of the key advances was the high-order smoothness analysis of \citet{DFG2021}, which established $O(n\log d_i\log^4 T)$ individual regret for Optimistic Hedge. A complementary direction sought faster rates by changing the learning rule. Lifting and adaptive learning speeds yield logarithmic external regret \citep{FarinaEtAl2022,SoleymaniEtAl2025,SoleymaniEtAl2025Fast}, while implicit updates and higher-order predictions yield constant regret at every horizon \citep{LiuFarinaOzdaglar2026,AbbadiLarakiMertikopoulos2026}. These results establish strong guarantees for modified dynamics, but leave open how far plain Optimistic Hedge can go. We show that Optimistic Hedge already can achieve logarithmic regret without any modifications.

\paragraph{The high-order smoothness framework.}
\citet{DFG2021} pioneered a framework that analyzes the high-order smoothness of loss sequences under self-play and established the first polylogarithmic individual regret bound for Optimistic Hedge in general-sum games. Their analysis first bounds high-order finite differences of losses, obtained by repeatedly subtracting consecutive terms. They then transfer these bounds one order at a time to first differences---the algorithm's prediction errors---while accounting for changing action probabilities. We build on this framework and analyze probability-weighted differences between losses. Our sharper high-order difference bound and direct comparison of high- and first-order differences together yield logarithmic regret.

\begin{table*}[!t]
\centering
\caption{Individual external regret under simultaneous self-play and exact expected loss-vector feedback; $d=\max_i d_i$. The result of \citet{ChenPeng2020} assumes two players; the others allow $n$ players. \emph{Single loop}: closed-form updates without inner optimization or root finding, excluding feedback cost. The OFTRL/OMD check refers to their entropy-regularized Optimistic Hedge instantiation.}
\label{tab:regret-comparison}
{
\setlength{\tabcolsep}{3pt}
\renewcommand{\arraystretch}{1.12}
\begin{tabular*}{\textwidth}{@{\extracolsep{\fill}}llcl@{}}
\toprule
Reference & Algorithm & Single loop & Regret bound \\
\midrule
\citet{SALS2015} & OFTRL / OMD & \cmark & $O\!\left(\sqrt n\log d\,T^{1/4}\right)$ \\
\citet{ChenPeng2020} & Optimistic Hedge & \cmark & $O\!\left(\log^{5/6}d\,T^{1/6}\right)$ \\
\citet{DFG2021} & Optimistic Hedge & \cmark & $O\!\left(n\log d\log^4T\right)$ \\
\citet{FarinaEtAl2022} & LRL-OFTRL & \xmark & $O(nd\log T)$ \\
\citet{SoleymaniEtAl2025,SoleymaniEtAl2025Fast} & COMWU & \xmark & $O\!\left(n\log^2d\log T\right)$ \\
\citet{LiuFarinaOzdaglar2026} & ECHO-OFTRL & \xmark & $O\!\left(n^{21}\log^4d\right)$ \\
\citet{AbbadiLarakiMertikopoulos2026} & HOOD & \xmark & $O\!\left(n^3\log^2d\right)$ \\
\textbf{This work} & Optimistic Hedge & \cmark & $O\!\left(\sqrt n\log d\log T\right)$ \\
\bottomrule
\end{tabular*}}
\end{table*}

\section{Preliminaries}
\label{sec:setup}

We study repeated play in finite general-sum games within the classical framework of no-regret learning \citep{Blackwell1956,Hannan1957,FudenbergLevine1998,CBL2006}.

\paragraph{Notation.}
For an integer $d\ge1$, define $[d]:=\{1,\ldots,d\}$ and the simplex $\Delta_d:=\{x\in\R_+^d:\sum_a x_a=1\}$. An $n$-player game has action sets $A_i=[d_i]$, joint action space $A:=\prod_{j\in[n]}A_j$, opponents' action space $A_{-i}:=\prod_{j\ne i}A_j$, and losses $\ell_i:A\to[0,1]$. Strategies $x_i\in\Delta_{d_i}$ form the profile $x=(x_j)_{j\in[n]}$, opponents' profile $x_{-i}=(x_j)_{j\ne i}$ and product distribution $\bigotimes_jx_j$ on $A$.

We also use $\ell_i:\prod_{j\ne i}\Delta_{d_j}\to[0,1]^{d_i}$ for the expected loss map. Unlike the scalar loss $\ell_i(a)$ at a pure action profile, the vector $\ell_i(x_{-i})$ is multilinear in the opponents' mixed strategies:
\begin{equation}
[\ell_i(x_{-i})]_{a_i}:=\sum_{a_{-i}\in A_{-i}}\ell_i(a_i,a_{-i})\prod_{j\ne i}x_{j,a_j}.
\label{eq:loss-vector}
\end{equation}
For a horizon $T\ge1$, players simultaneously choose a strategy profile $x^t$ at each round $t\in[T]$. Player $i$ incurs expected loss $\ip{x_i^t}{\ell_i^t}$ and observes the exact loss vector $\ell_i^t:=\ell_i(x_{-i}^t)$, using only its own observed vectors to update. Feedback is noiseless, and no additional game structure is assumed.

\begin{definition}[External regret]
Player $i$'s external regret against the best fixed action in hindsight is
\begin{equation}
\Reg_i(T):=\sum_{t=1}^T\ip{x_i^t}{\ell_i^t}
-\min_{a_i\in A_i}\sum_{t=1}^T\ell_{i,a_i}^t.
\label{eq:regret}
\end{equation}
\end{definition}

Regret $\Reg_i(T)$ may be negative, since adaptive play can outperform every fixed action. We seek bounds on each $\Reg_i(T)$, rather than only their sum.

While Nash equilibrium is one of the most widely used solution concepts in games \citep{Nash1951}, computing one is PPAD-complete even for two-player general-sum games \citep{DaskalakisGoldbergPapadimitriou2009,ChenDengTeng2009}. We therefore consider a standard relaxation, coarse correlated equilibrium (CCE) \citep{MoulinVial1978}, which permits correlated plays.

\begin{definition}[Coarse correlated equilibrium]
For $\varepsilon\ge0$, a distribution $\mu$ over $A$ is an $\varepsilon$-coarse correlated equilibrium ($\varepsilon$-CCE) if
\begin{equation}
\operatorname{Gap}_{\rm CCE}(\mu)
:=\max_{i\in[n]}\max_{a_i'\in A_i}
\E_{a\sim\mu}\!\left[\ell_i(a)-\ell_i(a_i',a_{-i})\right]\le\varepsilon.
\label{eq:cce-gap-definition}
\end{equation}
\end{definition}

\Needspace{6\baselineskip}
No-regret learning provides a direct route to CCE: if every player has sublinear external regret, average play approaches the CCE set \citep{HartMasColell2001,CBL2006}. Specifically, define the time-averaged joint distribution $\widehat\mu_T:=T^{-1}\sum_{t=1}^T\bigotimes_{j\in[n]}x_j^t$. It selects a common round uniformly, then samples from that round's product distribution. By multilinearity,

\begin{equation}
\operatorname{Gap}_{\rm CCE}(\widehat\mu_T)=T^{-1}\max_i\Reg_i(T).
\label{eq:cce-generic}
\end{equation}
The joint distribution $\widehat\mu_T$ averages products, not the individual strategies $x_i^t$.

\paragraph{Optimistic Hedge.}
Hedge is an instance of the follow-the-regularized-leader (FTRL) algorithm with negative entropy regularization \citep{FreundSchapire1997,FreundSchapire1999}. Optimistic Hedge is a variant of Hedge that uses the most recent loss vector $\ell_i^t$ to predict $\ell_i^{t+1}$ \citep{RakhlinSridharan2013,RakhlinSridharan2013Predictable}.

Starting from the uniform strategy $x_i^1=\1_{d_i}/d_i$ and initial prediction $\ell_i^0=0$, all players use a common constant step size $\eta>0$, chosen before play. For $t\ge1$, the variational update is
\begin{equation}
x_i^{t+1}
=\arg\min_{x_i\in\Delta_{d_i}}
\left\{\ip{x_i}{\sum_{s=1}^t\ell_i^s+\ell_i^t}
+\frac1\eta\sum_{a\in A_i}x_{i,a}\log x_{i,a}\right\}.
\label{eq:optimistic-ftrl}
\end{equation}
The entropy regularizer in \cref{eq:optimistic-ftrl} admits an explicit solution. Eliminating cumulative losses between consecutive rounds gives
\begin{equation}
x_{i,a}^{t+1}
=\frac{x_{i,a}^t\exp\{-\eta(2\ell_{i,a}^t-\ell_{i,a}^{t-1})\}}
{\sum_{b\in A_i}x_{i,b}^t\exp\{-\eta(2\ell_{i,b}^t-\ell_{i,b}^{t-1})\}},
\qquad a\in A_i.
\label{eq:omwu}
\end{equation}
With suitable step sizes, Optimistic Hedge improves the adversarial $O(\sqrt T)$ regret bound to $O(T^{1/4})$ \citep{SALS2015} and $O(\log^4T)$ \citep{DFG2021}. Our main result in the next section sharpens these guarantees to $O(\log T)$.

\section{Main result}
\label{sec:main-results}

The following result shows that plain Optimistic Hedge with a horizon-tuned constant step size can achieve a logarithmic individual regret bound in finite general-sum games.

\begin{restatable}[Parameterized individual-regret bound]{theorem}{mainResultRestated}
\label{thm:explicit}
Under the setup of Section~\ref{sec:setup}, there exist universal constants \(c_0,c_1,C>0\) such that, for every player \(i\), integer horizon \(T\ge1\), and \mbox{\(0<\eta\le c_0/\sqrt n\)},
\begin{equation}
\Reg_i(T)\le\frac{\log d_i}{\eta}+C
+C\eta T\exp\!\left(-\frac{c_1}{\sqrt n\eta}\right).
\label{eq:explicit-regret}
\end{equation}
\end{restatable}

The usual entropy bound leaves an $O(\eta T)$ term for bounded adversarial losses \citep[Lemma~4.1]{DFG2021}. \Cref{thm:explicit} replaces this by $O(\eta T\exp\{-c_1/(\sqrt n\eta)\})$, up to a bounded additive term. Hence, choosing $1/(\sqrt n\eta)$ to be a sufficiently large multiple of $\log(T+2)$ keeps the remainder bounded and gives the following logarithmic regret bound.

\begin{restatable}[Dimension-explicit logarithmic regret]{corollary}{logRegretRestated}
\label{cor:logarithmic}
There are universal constants $\bar c,C>0$ such that, for every player $i$ and integer horizon $T\ge1$, choosing the common constant step size \(\eta_T=\bar c/[\sqrt n\log(T+2)]\) before round~\(1\) gives
\begin{equation}
\Reg_i(T)\le C\sqrt n\,\log d_i\,\log(T+2).
\label{eq:log-regret}
\end{equation}
\end{restatable}

\Cref{cor:logarithmic} improves the $O(n\log d_i\log^4T)$ bound of \citet{DFG2021} in two respects: the horizon dependence decreases from $\log^4T$ to $\log T$, and the player dependence from $n$ to $\sqrt n$. The action dependence remains logarithmic. This improvement comes from the larger admissible step size $\eta=\Theta(1/[\sqrt n\log T])$, rather than $\Theta(1/[n\log^4T])$ in \citet{DFG2021}.

Recent work obtains logarithmic or constant external regret and sublogarithmic swap regret through modified regularization, adaptive step sizes, or higher-order predictions \citep{FarinaEtAl2022,SoleymaniEtAl2025,SoleymaniEtAl2025Fast,Tsuchiya2026Hybrid,LiuFarinaOzdaglar2026,AbbadiLarakiMertikopoulos2026}. The lifted and cautious methods in \cref{tab:regret-comparison}, however, require inner optimization or scalar root finding. \Cref{cor:logarithmic}, on the other hand, attains a logarithmic regret bound with the closed-form update rule in \cref{eq:omwu}.

\Needspace{8\baselineskip}
The following corollary shows that the improved individual regret bound yields a faster CCE approximation. It follows directly from \cref{cor:logarithmic} and the regret-to-CCE identity \cref{eq:cce-generic}.
\begin{restatable}[Coarse-correlated-equilibrium convergence]{corollary}{cceRestated}
\label{cor:cce}
Let \(d:=\max_i d_i\). Under the step-size choice of Corollary~\ref{cor:logarithmic}, there is a universal constant $C>0$ such that, for every integer $T\ge1$,
\begin{equation}
\operatorname{Gap}_{\rm CCE}(\widehat\mu_T)
\le \frac{C\sqrt n\log d\log(T+2)}{T}.
\label{eq:cce-rate}
\end{equation}
\end{restatable}

\Cref{cor:cce} improves the CCE approximation rate of Optimistic Hedge from $O(n\log d\log^4T/T)$ to $O(\sqrt n\log d\log T/T)$ \citep{DFG2021}.

\section{Proof Overview}
\label{sec:analysis}

We follow the approach of \citet{DFG2021} in bounding regret through high-order smoothness of self-play. A centered-logit recurrence gives sharper difference bounds, which together with fixed-norm finite-difference interpolation produce a larger admissible step size and logarithmic regret.

\subsection{Reduction to a First-Order Difference Estimate}

For a mixed strategy $x\in\Delta_d$ and a vector of action losses $v\in\R^d$, define the weighted variance $\Var_x(v):=\sum_{a=1}^d x_a(v_a-\ip{x}{v})^2$. Optimistic regret bounds charge prediction error and subtract stability \citep{RakhlinSridharan2013,RakhlinSridharan2013Predictable}, as formalized by the regret bounded by variation in utilities (RVU) framework \citep{SALS2015}. We use the variance refinement of \citet[Lemma~4.1]{DFG2021}, specialized to a small constant step size.

\begin{restatable}[Variance-based RVU bound]{lemma}{regretSimpleRestated}
\label{lem:regret-simple}
There is a universal $\bar\eta>0$ such that, for every player $i$, integer horizon $T\ge1$, and $0<\eta\le\bar\eta$,
\begin{equation}
\Reg_i(T)
\le\frac{\log d_i}{\eta}
 +\frac{2\eta}{3}\sum_{t=1}^T\Var_{x_i^t}(\ell_i^t-\ell_i^{t-1})
 -\frac{\eta}{3}\sum_{t=1}^T\Var_{x_i^t}(\ell_i^{t-1}).
\label{eq:regret-simple}
\end{equation}
\end{restatable}

\Cref{lem:regret-simple} directs us to bound the prediction-error variance sum by half the stability variance sum, up to a small remainder. A sup-norm bound, however, discards the weights $x_i^t$, while retaining them inside the variance leaves a weighting that changes across rounds. The identity $\Var_x(v)=\sum_{a<b}x_ax_b(v_a-v_b)^2$ suggests absorbing $\sqrt{x_ax_b}$ into each pairwise loss difference to work in a fixed Euclidean norm. This motivates us to define the weighted-gap vector $h_i^t\in\R^{\binom{d_i}{2}}$, for $a<b$, by
\begin{equation}
[h_i^t]_{ab}
:=\sqrt{x_{i,a}^t x_{i,b}^t}
  \bigl(\ell_{i,a}^{t-1}-\ell_{i,b}^{t-1}\bigr).
\label{eq:h}
\end{equation}
The squared norm $\norm{h_i^t}_2^2$ already equals the stability variance. To control prediction-error variance using the same sequence, we must relate it to $\norm{h_i^{t+1}-h_i^t}_2^2$. Prediction error compares consecutive losses under the same weights $x_i^t$, whereas $h_i^{t+1}-h_i^t$ changes both the losses and their weights. The following identity isolates this weight change. For each $a<b$, write $w^t:=\sqrt{x_{i,a}^t x_{i,b}^t}$. Then
\[
w^t\bigl[(\ell_{i,a}^t-\ell_{i,b}^t)-(\ell_{i,a}^{t-1}-\ell_{i,b}^{t-1})\bigr]
=[h_i^{t+1}-h_i^t]_{ab}
 +\left(\frac{w^t}{w^{t+1}}-1\right)[h_i^{t+1}]_{ab}.
\]
The second term corrects for the changing weights. Optimistic Hedge gives $w^t/w^{t+1}-1=O(\eta)$, so the squared correction, summed over action pairs, is at most $O(\eta^2)\norm{h_i^{t+1}}_2^2$. We square this identity, bound the cross term, and sum over action pairs and rounds to obtain the following reduction.

\begin{restatable}[Weighted-gap reduction]{lemma}{weightedGapRestated}
\label{lem:weighted-gap}
There are universal constants \(C>0\) and \(\eta_0>0\) such that, for every player $i$, integer horizon $T\ge1$, and \(0<\eta\le\eta_0\),
\begin{align}
\sum_{t=1}^T\Var_{x_i^t}(\ell_i^{t-1})&=\sum_{t=1}^T\norm{h_i^t}_2^2,
\label{eq:S-gap}\\
\sum_{t=1}^T\Var_{x_i^t}(\ell_i^t-\ell_i^{t-1})&\le\frac32\sum_{t=1}^T\norm{h_i^{t+1}-h_i^t}_2^2
 +C\eta^2\sum_{t=1}^T\norm{h_i^t}_2^2+C\eta^2.
\label{eq:V-gap}
\end{align}
\end{restatable}

\Cref{lem:regret-simple,lem:weighted-gap} reduce bounding regret to comparing the two sums $\sum_{t=1}^T\norm{h_i^{t+1}-h_i^t}_2^2$ and $\sum_{t=1}^{T+1}\norm{h_i^t}_2^2$; the following lemma provides that key comparison.

\Needspace{5\baselineskip}
\begin{lemma}[First-order difference estimate]
\label{lem:explicit-temporal}
There are universal constants \(c,C>0\) such that, for every player $i$, integer horizon $T\ge1$, and \(0<\eta\le c/\sqrt n\),
\begin{equation}
\sum_{t=1}^T\norm{h_i^{t+1}-h_i^t}_2^2
\le\frac14\sum_{t=1}^{T+1}\norm{h_i^t}_2^2
 +C\left(\frac1{\sqrt n\,\eta}
 +T e^{-c/(\sqrt n\eta)}\right).
\label{eq:explicit-temporal}
\end{equation}
\end{lemma}

Together with \cref{lem:weighted-gap}, this estimate controls the variance terms in \cref{eq:regret-simple} and proves \cref{thm:explicit}. We derive this key lemma in the following subsections.

\subsection{Reduction to a High-Order Difference}

For a sequence $y=(y^t)_{t\in\mathbb Z}$, define finite differences by $\Delta^0y^t:=y^t$, $\Delta y^t:=y^{t+1}-y^t$, and $\Delta^{k+1}y^t:=\Delta(\Delta^ky^t)$, with squared sequence norm $\norm y_{\ell_2}^2:=\sum_t\norm{y^t}^2$ for finite support. \Cref{lem:weighted-gap}~motivates bounding the squared first differences $\norm{\Delta h_i^t}_2^2$ by the squared norms $\norm{h_i^t}_2^2$. The following lemma provides this comparison with an additional high-order difference term.

\begin{restatable}[Finite-difference interpolation]{lemma}{directFourierRestated}
\label{lem:fourier-tools}
For every finitely supported sequence $y$ in a finite-dimensional Hilbert space, integer \(m\ge1\), and \(\theta\in(0,1)\),
\begin{equation}
\norm{\Delta y}_{\ell_2}^2
\le \theta\norm{y}_{\ell_2}^2
 +\theta^{1-m}\norm{\Delta^m y}_{\ell_2}^2.
\label{eq:direct-fourier}
\end{equation}
\end{restatable}

The proof applies Parseval's identity and splits frequencies $\omega\in[-\pi,\pi]$ according to whether the condition $|e^{i\omega}-1|^2\le\theta$ holds or not. The first term $\theta\norm y_{\ell_2}^2$ bounds the low-frequency contribution to $\norm{\Delta y}_{\ell_2}^2$, and the high-order term $\theta^{1-m}\norm{\Delta^m y}_{\ell_2}^2$ bounds the contribution from the remaining frequencies. Hence, \cref{lem:fourier-tools} suggests proving \cref{lem:explicit-temporal} by deriving a sharp bound on $\norm{\Delta^m h_i^t}_2$.

\subsection{The Centered-Logit Recurrence and High-Order Difference Estimate}

To bound the high-order differences $\Delta^kh_i^t$ of the weighted gaps, we express Optimistic Hedge through a centered-logit state $z^t$. Derivative bounds for the centered loss map $F$ and weighted-gap map $H_i$, combined with the centered-logit recurrence and a chain rule, then give the required estimates.

\paragraph{Centered-logits as a state.}
Since each $x_i^t$ satisfies $x_{i,a}^t>0$ for all $a\in A_i$, we can define the \emph{centered logits} $z_i^t:=P_i\log x_i^t$, where the centering map $P_iv:=v-d_i^{-1}(\sum_bv_b)\1_{d_i}$ subtracts the coordinate mean. The logits $z_i^t$ encode relative probabilities: $z_{i,a}^t-z_{i,b}^t=\log(x_{i,a}^t/x_{i,b}^t)$, and the softmax map $[\softmax(v)]_a:=e^{v_a}/\sum_b e^{v_b}$ recovers the strategy $x_i^t=\softmax(z_i^t)$. We write $\osc(v):=\max_av_a-\min_av_a$ for the coordinate span and define the centered-logit spaces and product norm by
\begin{equation}
E_i:=\{z_i\in\R^{d_i}:\1_{d_i}^{\top}z_i=0\},
\qquad E:=\prod_{i=1}^n E_i,
\qquad
\norm{z}_E:=\left(\sum_{i=1}^n\osc(z_i)^2\right)^{1/2}.
\label{eq:logit-product-norm}
\end{equation}
\paragraph{Centered-logit recurrence.}
Define the centered loss map $F:E\to E$, which maps centered logits to negated centered loss vectors, and the weighted loss-gap map $H_i:E\times E\to\R^{\binom{d_i}{2}}$, which maps two centered-logit profiles to weighted gaps, by
\begin{align}
F_i(z)&:=[F(z)]_i=-P_i\ell_i\bigl((\softmax(z_j))_{j\ne i}\bigr),
\label{eq:F-main}\\
[H_i(z,z')]_{ab}
&:=\sqrt{[\softmax(z_i)]_a[\softmax(z_i)]_b}
 \bigl([F(z')]_{i,b}-[F(z')]_{i,a}\bigr),
\qquad a<b.
\label{eq:H-main}
\end{align}
The map $H_i(z,z')$ uses $z$ for the probability weights and $z'$ for the losses. For the joint centered-logit state $z^t:=(z_i^t)_{i=1}^n$, \cref{eq:omwu} gives $z^1=0$, $z^2=2\eta F(z^1)$, and, for $t\ge2$,
\begin{equation}
z^{t+1}=z^t+\eta\bigl(2F(z^t)-F(z^{t-1})\bigr),
\qquad h_i^t=H_i(z^t,z^{t-1}).
\label{eq:main-logit-gap-representation}
\end{equation}
\paragraph{High-order difference.}
Applying $k-1$ finite differences to the increment $\Delta z^t$ in the centered-logit recurrence gives, for $k\ge1$ and $t\ge2$,
\begin{equation}
\Delta^kz^t=\eta\bigl(2\Delta^{k-1}F(z^t)-\Delta^{k-1}F(z^{t-1})\bigr).
\label{eq:logit-difference-recurrence}
\end{equation}
Hence loss differences of order $k-1$ control centered-logit differences of order $k$, with a factor $\eta$. We will transfer this control through $H_i$ and prove the following estimate:

\begin{restatable}[High-order difference estimate of the weighted gaps]{lemma}{highOrderGapRestated}
\label{lem:high-order-gap}
There are universal constants $C,R>0$ such that, for every player $i$ and integer $k\ge1$,
\begin{equation}
\sup_{t\ge k+2}\norm{\Delta^kh_i^t}_2
\le C(R\sqrt n\,\eta)^k k!.
\label{eq:high-order-gap}
\end{equation}
\end{restatable}

Notably, the weighted-gap estimate in \cref{eq:high-order-gap} becomes exponentially small at a suitable difference order. To see this, write $\delta:=\sqrt n\eta$ and choose the difference order $m=\lfloor a/\delta\rfloor$, with $a>0$, $Ra<1$, and $\delta\le a/2$. Then $m\ge a/(2\delta)$ and
\[
\sup_{t\ge m+2}\norm{\Delta^m h_i^t}_2
\le C(R\delta)^m m!
\le C(R\delta m)^m
\le C(Ra)^m
\le C e^{-c/\delta}.
\]
Using $m!\le m^m$ and choosing $a$ small enough, this decay offsets the interpolation factor $\theta^{1-m}$. Together with the finite-difference interpolation in \cref{lem:fourier-tools}, \cref{lem:high-order-gap} yields \cref{lem:explicit-temporal}. The remaining task is therefore to prove the weighted-gap bound in \cref{lem:high-order-gap}.

\subsection{Proof of the High-Order Difference Estimate}
\label{sec:high-order-proof}

We first bound the high-order differences $\norm{\Delta^k z^t}_E$ of the centered logits using derivative bounds and a finite-difference chain rule. We then transfer these bounds through the weighted-gap map $H_i$ to prove \cref{lem:high-order-gap}. For a smooth map $\Psi$, write $D^r\Psi(x)$ for its $r$th Fr\'echet derivative at $x$ and $\norm{D^r\Psi(x)}$ for its induced operator norm; see \appref{app:frechet-derivatives} for the rigorous definition. For a real formal power series $U(s)$, let $[s^k]U$ denote the coefficient of $s^k$.

\paragraph{Step 1: control high-order derivatives $\norm{D^rF}_E$.}
We bound the derivatives $D^rF$ through coordinate differences of the centered loss map $F_i$. Write $F_{i,a}$ for its action-$a$ coordinate. For $z\in E$ and directions $v^1,\ldots,v^r\in E$, the product norm $\norm{\cdot}_E$ from \cref{eq:logit-product-norm} gives
\[
\norm{D^rF(z)[v^1,\ldots,v^r]}_E
\le\sqrt n\,\max_i\max_{a,b\in A_i}
\abs{D^r(F_{i,a}-F_{i,b})(z)[v^1,\ldots,v^r]}.
\]
Now suppose we draw the opponents' actions $\mathsf A_j$ independently from $\softmax(z_j)$ for $j\ne i$. Their joint action $\mathsf A_{-i}$ determines the loss gap $g:=\ell_i(b,\mathsf A_{-i})-\ell_i(a,\mathsf A_{-i})$. For each direction $v^q$, let $v_j^q\in E_j$ be the perturbation to player $j$'s logits; evaluating these perturbations at the same opponent actions and centering their sum gives $S_q:=\sum_{j\ne i}\bigl(v_j^q(\mathsf A_j)-\E v_j^q(\mathsf A_j)\bigr)$, and the softmax in $F_i$ gives the following joint-cumulant representation
\[
D^r(F_{i,a}-F_{i,b})(z)[v^1,\ldots,v^r]
=\kappa(g,S_1,\ldots,S_r).
\]
Here $\kappa(g,S_1,\ldots,S_r)$ is the joint cumulant of $g,S_1,\ldots,S_r$, defined by
\[
\kappa(Y_1,\ldots,Y_m)
:=\sum_{\pi\in\Pi_m}(-1)^{|\pi|-1}(|\pi|-1)!
  \prod_{B\in\pi}\E\prod_{j\in B}Y_j,
\]
for real random variables $Y_1,\ldots,Y_m$ with finite mixed moments, where $\Pi_m$ is the set of partitions of $[m]$ into nonempty blocks \citep{LeonovShiryaev1959}.

The joint-cumulant bound in \cref{lem:cumulant}, boundedness of $g$, and moment bounds for the centered sums $S_q$ in terms of $\norm{v^q}_E$ yield the following derivative estimate; \appref{app:loss-derivative-proof} gives the proof.

\begin{restatable}[Factorial derivatives of the centered loss map]{lemma}{lossResponseDerivativeRestated}
\label{lem:F-derivative-bounds}
There are universal constants $A,R>0$ such that every integer $r\ge1$, player $i$, point $z\in E$, and directions $v^1,\ldots,v^r\in E$ satisfy
\begin{align}
\osc\!\left(D^rF_i(z)[v^1,\ldots,v^r]\right)
&\le AR^rr!\prod_{q=1}^r\norm{v^q}_E,
\label{eq:Fi-derivative}\\
\norm{D^rF(z)[v^1,\ldots,v^r]}_E
&\le A\sqrt n\,R^rr!\prod_{q=1}^r\norm{v^q}_E.
\label{eq:F-derivative}
\end{align}
\end{restatable}

\paragraph{Step 2: transfer to high-order differences $\norm{\Delta^kz^t}_E$.}
By the high-order difference identity given by \cref{eq:logit-difference-recurrence}, bounding the finite difference $\Delta^kz^t$ reduces to bounding order-$(k-1)$ differences of $F(z^t)$. The following chain rule uses derivatives of $F$ and lower-order logit differences to control these terms, extending the chain rule of \citet[Lemma~4.5]{DFG2021} to normed spaces.

\begin{restatable}[Finite-difference chain rule]{lemma}{differenceChainRuleRestated}
\label{lem:difference-chain-rule}
Let $X,Y$ be finite-dimensional normed spaces, $(u^t)_{t\ge1}$ a sequence in $X$, and $\Psi:X\to Y$ smooth. Suppose $a_r:=\sup_{t\ge r+1}\norm{\Delta^ru^t}/r!$ and $c_p:=\sup_{x\in X}\norm{D^p\Psi(x)}/p!$ are finite for all integers $r,p\ge1$. Then, for every $k\ge1$,
\begin{equation}
\sup_{t\ge k+1}\frac{\norm{\Delta^k\Psi(u^t)}}{k!}
\le[s^k]\sum_{p=1}^kc_p
 \left(\sum_{r=1}^ka_rs^r\right)^p.
\label{eq:difference-chain-rule}
\end{equation}
\end{restatable}

Like the usual chain rule, \cref{lem:difference-chain-rule} controls finite differences of the composition $\Psi(u^t)$ using derivatives of $\Psi$ and finite differences of $u^t$. Here $a_r$ bounds the $r$th input difference normalized by $r!$, and $c_p$ bounds the $p$th derivative normalized by $p!$. Intuitively, repeated differentiation acts on the outer map and transfers to its input derivatives; \cref{eq:difference-chain-rule} bounds this chained contribution.

Together with the high-order difference identity in \cref{eq:logit-difference-recurrence} and the derivative bounds in \cref{lem:F-derivative-bounds}, the chain rule with $\Psi=F$ and $u^t=z^t$ at order $k-1$ bounds order-$k$ centered-logit differences $\Delta^kz^t$ using orders $1,\ldots,k-1$. Starting from $\norm{\Delta z^t}_E\le3\sqrt n\eta$, the following lemma controls how these bounds grow with $k$.

\begin{restatable}[High-order differences of the centered-logit state]{lemma}{highOrderLogitRestated}
\label{lem:high-order-logit}
For the trajectory $z^t$ generated by the centered-logit recurrence, there are universal constants $C_z,R_z>0$ such that, for every integer $k\ge1$,
\begin{equation}
\sup_{t\ge k+1}\norm{\Delta^kz^t}_E
\le C_z(R_z\sqrt n\,\eta)^k k!.
\label{eq:z-high-order}
\end{equation}
\end{restatable}

The proof works with $a_k:=\sup_{t\ge k+1}\norm{\Delta^kz^t}_E/k!$ and seeks a geometric bound in $k$. The derivative estimates and the chain rule turn the logit recurrence into a polynomial upper bound for $a_k$ in terms of $a_1,\ldots,a_{k-1}$ for $k\ge2$. As the chain rule involves only sums and products of norm bounds, this polynomial upper bound has nonnegative coefficients. As a result, replacing each $a_j$ in the polynomial by a number $w_j\ge a_j$, for $1\le j<k$, gives another upper bound on $a_k$. Hence, we can start from the initial upper bound $w_1=C\sqrt n\eta\ge a_1$ for some constant $C$, and define the upper bound $w_k$ by evaluating the polynomial at previous bounds $w_1,\ldots,w_{k-1}$. Computing the power series $\sum_{k\ge1}w_ks^k$ explicitly gives $a_k\le w_k\le C_z(R_z\sqrt n\,\eta)^k$; restoring $k!$ proves the lemma.

The bound has one factor of $\sqrt n\eta$ per difference order: $\eta$ comes from the update and $\sqrt n$ from combining playerwise derivative bounds in the product norm. After division by $k!$, the remaining growth is geometric with constants independent of the order. This permits a suitable choice of finite-difference order $m\asymp(\sqrt n\eta)^{-1}$ that yields an exponentially small bound on high-order differences.

\paragraph{Step 3: transfer to high-order differences $\norm{\Delta^kh_i^t}_2$.}
Recall the identity $h_i^{t+1}=H_i(z^{t+1},z^t)$; hence, we can transfer the estimate of centered-logits $\Delta^kz^t$ in \cref{lem:high-order-logit} to an estimate of weighted gaps $\Delta^kh_i^t$ through the map $H_i$. \Cref{lem:H-derivative-bounds} supplies a derivative bound of $H_i$ for this argument.

\begin{restatable}[Factorial derivatives of the weighted loss-gap map]{lemma}{weightedGapDerivativeRestated}
\label{lem:H-derivative-bounds}
For the weighted-gap map $H_i:E\times E\to\R^{\binom{d_i}{2}}$, equip the input space $E\times E$ with $\norm{(p,q)}=(\norm p_E^2+\norm q_E^2)^{1/2}$ and the output space with the Euclidean norm on unordered action pairs. Then, there are universal constants $A_H,R_H>0$ such that, for every player $i$ and integer $r\ge0$,
\begin{equation}
\sup_{z,z'\in E}\norm{D^rH_i(z,z')}
\le A_HR_H^rr!.
\label{eq:H-derivative}
\end{equation}
\end{restatable}

For a mixed strategy $x_i=\softmax(z_i)$, differentiating the exponential formula for $\sqrt{x_{i,a}}$ leaves $\sqrt{x_{i,a}}$ as a multiplier. Cumulant bounds in \cref{lem:cumulant} and \cref{lem:F-derivative-bounds} then give factorial derivative bounds for $H_i$ weighted by $\sqrt{x_{i,a}x_{i,b}}$, which can be bounded as $\sum_{a<b}x_{i,a}x_{i,b}\le1/2$.

Now, we are ready to prove \cref{lem:high-order-gap} by combining \cref{lem:difference-chain-rule}, \cref{lem:high-order-logit}, and \cref{lem:H-derivative-bounds}.
\par
\begin{proof}[Proof of \Cref{lem:high-order-gap}]
Write $\delta:=\sqrt n\eta$. For the paired centered-logit state $(z^{t+1},z^t)$, \cref{lem:high-order-logit} and the product norm in \cref{lem:H-derivative-bounds} give $\sup_{t\ge k+1}\norm{\Delta^k(z^{t+1},z^t)}/k!\le C_z(R_z\delta)^k$, after enlarging $C_z$ by $\sqrt2$. Applying \cref{lem:difference-chain-rule,lem:H-derivative-bounds} to $h_i^{t+1}=H_i(z^{t+1},z^t)$ yields
\[
\sup_{t\ge k+1}\frac{\norm{\Delta^kh_i^{t+1}}_2}{k!}
\le[s^k]A_H\frac{R_HC_zR_z\delta s}{1-(1+R_HC_z)R_z\delta s}
\le C(R\delta)^k,
\]
The chain rule sums the geometric series of input difference bounds. Multiplying by $k!$ and shifting to $t+1\ge k+2$ proves \cref{lem:high-order-gap}.
\end{proof}

\subsection{Proof of the First-Order Difference Estimate}
\label{sec:temporal-estimate}

\Cref{lem:high-order-gap} successfully controls the high-order difference of weighted gaps $h_i^t$. However, the finite-difference interpolation in \cref{lem:fourier-tools} requires a finitely supported sequence. The following cutoff lemma transforms a given sequence $u^t$ into a finitely supported one $y^t$ while preserving the smoothness of the original sequence $u^t$ within a small error. We defer the proof to \appref{app:finite-horizon}.

\begin{restatable}[Finite-support cutoff]{lemma}{finiteSupportCutoffRestated}
\label{lem:boundary-extension}
Fix $B,C_0,R_0,\rho>0$ and integers $m\ge1$, $L\ge2$. Let $(u^t)_{t\ge1}$ be a sequence in a finite-dimensional Euclidean space satisfying
\begin{equation}
\sup_{t\ge1}\norm{u^t}\le B,
\qquad
\sup_{t\ge k+2}\norm{\Delta^ku^t}\le C_0(R_0\rho)^kk!
\quad(1\le k\le m).
\label{eq:boundary-extension-assumption}
\end{equation}
There is a constant $C_L>0$, depending only on $B,C_0,L$, such that every integer horizon $T\ge1$ admits a finitely supported sequence $y=(y^t)_{t\in\mathbb Z}$ with
\begin{align}
\norm{\Delta y}_{\ell_2}^2
&\ge{\textstyle\sum_{t=1}^T}\norm{\Delta u^t}^2-C_Lm,
\label{eq:boundary-extension-difference}\\
\norm{y}_{\ell_2}^2
&\le{\textstyle\sum_{t=1}^{T+1}}\norm{u^t}^2+C_Lm,
\label{eq:boundary-extension-level}\\
\norm{\Delta^my}_{\ell_2}^2
&\le C_L(T+m)(2/L+R_0\rho m)^{2m}.
\label{eq:boundary-extension-high-order}
\end{align}
\end{restatable}

Intuitively, \cref{lem:boundary-extension} constructs a tapered, finitely supported version $y$ of the weighted gaps $h_i^t$. Now, we can apply finite-difference interpolation (\cref{lem:fourier-tools}) to $y$, then transfer the comparison back to $h_i^t$, and prove \cref{lem:explicit-temporal}. With a sufficiently large $L$ and a suitable $m\asymp(\sqrt n\eta)^{-1}$, this gives a $O(m)$ boundary error and a small $O((T+m)e^{-cm})$ contribution from the high-order differences.

\begin{proof}[Proof of \Cref{lem:explicit-temporal}]
Set $\delta:=\sqrt n\eta$ and $\theta:=1/4$, and let $C,R$ be the constants in \cref{lem:high-order-gap}. Choose a fixed integer $L\ge2$ sufficiently large and $a>0$ sufficiently small that $\gamma:=2/L+Ra$ satisfies $\gamma^2/\theta\le1/2$. For $\delta\le a/2$, set $m:=\lfloor a/\delta\rfloor$, so $a/(2\delta)\le m\le a/\delta$.

By $\norm{h_i^t}_2\le1/2$ and \cref{lem:high-order-gap}, \cref{lem:boundary-extension} applies with $u^t=h_i^t$ and $(B,C_0,R_0,\rho)=(1/2,C,R,\delta)$. It gives $\norm{\Delta^my}_{\ell_2}^2\le C_L(T+m)\gamma^{2m}$ since $R\delta m\le Ra$. Its two squared-norm comparisons and finite-difference interpolation (\cref{lem:fourier-tools}) yield
\begin{align*}
\sum_{t=1}^T\norm{\Delta h_i^t}_2^2
&\le\theta\sum_{t=1}^{T+1}\norm{h_i^t}_2^2
 +Cm+C(T+m)\theta^{1-m}\gamma^{2m}\\*
&\le\frac14\sum_{t=1}^{T+1}\norm{h_i^t}_2^2
 +C\left(\frac1\delta+T e^{-c'/\delta}\right).
\end{align*}
The last line uses $\theta^{1-m}\gamma^{2m}\le2^{-m}$, $m\le a/\delta$, and $2^{-m}\le e^{-c'/\delta}$ for $c':=a\log2/2$. Since $L$ is fixed, the constants are universal. Taking the stated $c\le\min\{a/2,c'\}$ proves \cref{lem:explicit-temporal}.
\end{proof}

\paragraph{Completing the regret bound.}
Combining \cref{lem:regret-simple,lem:weighted-gap,lem:explicit-temporal} proves \cref{thm:explicit}; choosing $\eta_T=\bar c/[\sqrt n\log(T+2)]$ gives \cref{cor:logarithmic}. \appref{app:main-result} supplies the complete regret calculation.

\section{Discussion}
\label{sec:discussion}
\paragraph{Implications.}
A notable recent line of work obtains logarithmic or constant regret bounds through modified learning dynamics \citep{FarinaEtAl2022,AbbadiLarakiMertikopoulos2026}. Our result shows that such modifications are unnecessary for logarithmic individual regret in general-sum games: under expected loss-vector feedback, Optimistic Hedge with a horizon-tuned constant step size already achieves this guarantee.

\paragraph{Improvement upon \citet{DFG2021}.}
Our larger constant step size $\eta=\Theta(1/[\sqrt n\log T])$, compared to the step size $\eta=\Theta(1/[n\log^4 T])$ in \citet{DFG2021}, comes from two improvements that work together. The factorial estimate in \cref{lem:high-order-gap} controls high-order differences with an exponentially small bound at a suitably chosen order $m\asymp(\sqrt n\eta)^{-1}$. Then, the weighted-gap representation lets finite-difference interpolation transfer this control directly to first differences in a fixed Euclidean norm, avoiding the repeated comparisons under changing weights in \citet{DFG2021}. The $\sqrt n$ factor, on the other hand, originates from combining playerwise derivative bounds in the product norm. \appref{app:dfg-comparison} details the estimates underlying these conditions.

\section{Conclusion}
In this work, we established the first logarithmic $O(\sqrt n\log d_i\log T)$ individual regret bound for Optimistic Hedge in general-sum games under expected loss-vector feedback. Our analysis transfers derivative bounds through a centered-logit recurrence to obtain a high-order difference bound for weighted loss gaps, and finite-difference interpolation yields a larger admissible step size. Whether Optimistic Hedge can achieve horizon-independent regret remains an open question.

\clearpage
\subsection*{AI use statement}
Generative AI tools assisted with mathematical exploration, including formulating and refining claims and hypotheses, developing the analytical framework and proof ingredients, drafting and revising proofs, and interpreting the resulting bounds. They also assisted with literature search and comparison, manuscript organization and exposition, LaTeX editing, and reference formatting. No datasets or empirical experiments are reported in this theoretical work. The author takes responsibility for the final claims, proofs, citations, and text, including material developed with AI assistance.

\subsection*{Reproducibility statement}
\Cref{sec:setup,sec:main-results} specify the setup, algorithm, and assumptions. Proofs of the theoretical results are provided in \cref{sec:analysis} and Appendices~\ref{app:main-result}--\ref{app:finite-horizon}.

\bibliography{references}

\clearpage
\begingroup
\renewcommand{\contentsname}{Appendix contents}
\setcounter{tocdepth}{-1}
\tableofcontents
\endgroup
\clearpage
\appendix
\addtocontents{toc}{\protect\setcounter{tocdepth}{2}}

\section{Proof of the main result}
\label{app:main-result}

We prove the regret bound for a constant step size, its logarithmic specialization, and the CCE guarantee in order.

\printrepeatedstatement{thm:explicit}{\mainResultRestated}

\begin{proof}
Choose $c_0\le1$ no larger than the universal thresholds in \cref{lem:regret-simple,lem:weighted-gap,lem:explicit-temporal} and small enough that $3/8+C\eta^2\le1/2$ whenever $\eta\le c_0/\sqrt n$. Since $\norm{h_i^{T+1}}_2^2=\Var_{x_i^{T+1}}(\ell_i^T)\le1/4$ for losses in $[0,1]$, \cref{lem:weighted-gap,lem:explicit-temporal} give
\begin{align*}
\sum_{t=1}^T\Var_{x_i^t}(\ell_i^t-\ell_i^{t-1})
&\le\left(\frac38+C\eta^2\right)\sum_{t=1}^T\norm{h_i^t}_2^2
 +C\left(1+\frac1{\sqrt n\eta}+T e^{-c/(\sqrt n\eta)}\right)\\
&\le\frac12\sum_{t=1}^T\norm{h_i^t}_2^2
 +C\left(\frac1{\sqrt n\eta}+T e^{-c/(\sqrt n\eta)}\right).
\end{align*}
The second line also uses $\sqrt n\eta\le c_0\le1$, so $1\le1/(\sqrt n\eta)$. Substitution into \cref{eq:regret-simple} yields
\begin{align*}
\Reg_i(T)
&\le\frac{\log d_i}{\eta}
 +\frac{2\eta}{3}\left[\frac12\sum_{t=1}^T\norm{h_i^t}_2^2
 +C\left(\frac1{\sqrt n\eta}+T e^{-c/(\sqrt n\eta)}\right)\right]\\*
&\hspace{1cm}-\frac{\eta}{3}\sum_{t=1}^T\norm{h_i^t}_2^2
 \le\frac{\log d_i}{\eta}+C+C\eta T e^{-c/(\sqrt n\eta)}.
\end{align*}
Taking $c_1$ to be the exponential constant supplied by \cref{lem:explicit-temporal} proves \cref{eq:explicit-regret}.
\end{proof}

\Cref{thm:explicit} leaves only $C\eta T e^{-c_1/(\sqrt n\eta)}$ to control. The horizon-tuned constant step size $\eta_T$ makes this remainder bounded; the entropy term $(\log d_i)/\eta_T$ then gives \cref{eq:log-regret}.

\printrepeatedstatement{cor:logarithmic}{\logRegretRestated}

\begin{proof}
Choose $\bar c>0$ so that $\bar c\le c_0\log 3$ and $c_1/\bar c\ge2$. Then $\eta_T\le c_0/\sqrt n$ for every $T\ge1$, and
\[
\exp\!\left(-\frac{c_1}{\sqrt n\eta_T}\right)
=(T+2)^{-c_1/\bar c}\le(T+2)^{-2}.
\]
Substituting $\eta_T=\bar c/[\sqrt n\log(T+2)]$ into \cref{eq:explicit-regret} gives
\[
\Reg_i(T)
\le \frac{\sqrt n\,\log d_i\,\log(T+2)}{\bar c}+C
+\frac{C\bar cT}{\sqrt n\log(T+2)(T+2)^2}.
\]
For $d_i=1$, regret is zero. For $d_i\ge2$, $\sqrt n\log d_i\log(T+2)\ge\log2\log3$, so the bounded terms are absorbed by enlarging the universal constant. This proves \cref{eq:log-regret}.
\end{proof}

The identity in \cref{eq:cce-generic} concerns $\widehat\mu_T$, the time-average of the joint distributions $\bigotimes_jx_j^t$, and is exact under the expected-vector model.

\printrepeatedstatement{cor-cce}{\cceRestated}

\begin{proof}
For any player \(i\) and fixed deviation \(a_i'\), multilinearity gives
\[
\E_{a\sim\widehat\mu_T}\ell_i(a)-\E_{a\sim\widehat\mu_T}\ell_i(a_i',a_{-i})
=\frac1T\sum_{t=1}^T\bigl(\ip{x_i^t}{\ell_i^t}-\ell_{i,a_i'}^t\bigr).
\]
Maximizing over \(a_i'\) and then over \(i\) proves \cref{eq:cce-generic}. Combining this identity with \cref{cor:logarithmic} and $d_i\le d$ proves \cref{eq:cce-rate}.
\end{proof}

\section{Preliminary lemmas}
\label{app:common}

We recall the variance-based RVU bound for Optimistic Hedge, then prove the weighted-gap representation used to compare its two variance sums.

\subsection{The variance-based RVU bound}

\paragraph{Soft minimum and Bregman divergence.}
Fix a player $i$, suppress its index, and fix $\eta>0$. For loss-score vectors $q,q'\in\R^{d_i}$, define the soft minimum $\Phi:\R^{d_i}\to\R$ and the Bregman divergence of $-\Phi$, denoted by $D_\Phi:\R^{d_i}\times\R^{d_i}\to\R_+$, by
\[
\Phi(q):=-\frac1\eta\log\sum_{a=1}^{d_i} e^{-\eta q_a},
\qquad
D_\Phi(q',q):=\Phi(q)+\ip{\nabla\Phi(q)}{q'-q}-\Phi(q').
\]
Here $\nabla\Phi(q)=\softmax(-\eta q)$.

We recover the variance refinement of the RVU bound in \citet[Lemma~4.1]{DFG2021}, following optimistic prediction-error analysis \citep{RakhlinSridharan2013Predictable,RakhlinSridharan2013} and the RVU framework \citep{SALS2015}. Regret contains a difference of divergence sums, so we need the following comparison with local variance from both above and below.

\begin{lemma}[Soft-min divergence and local variance]
\label{lem:softmin-variance}
Let $p=\nabla\Phi(q)$ and $\osc(v)\le R$. Then
\begin{equation}
\frac\eta2 e^{-\eta R}\Var_p(v)
\le D_\Phi(q+v,q)
\le\frac\eta2 e^{\eta R}\Var_p(v).
\label{eq:softmin-variance}
\end{equation}
\end{lemma}

\begin{proof}
We compare the weighted variance along the segment $q+sv$ with the weighted variance under $p$. Set $p_s:=\nabla\Phi(q+sv)$ for $0\le s\le1$. The Hessian identity and Taylor's integral remainder give
\begin{align}
-\nabla^2\Phi(q+sv)
&=\eta\bigl(\diag(p_s)-p_sp_s^\top\bigr),\notag\\
D_\Phi(q+v,q)
&=\eta\int_0^1(1-s)\Var_{p_s}(v)\,\dd s.
\label{eq:hessian-integral}
\end{align}
Thus $D_\Phi(q+v,q)$ averages $\Var_{p_s}(v)$ along the segment. To compare with $\Var_p(v)$, use $\Var_x(v)=\min_{c\in\R}\sum_a x_a(v_a-c)^2$ for a mixed strategy $x\in\Delta_{d_i}$. Then
\begin{align*}
e^{-\eta sR}
&\le\frac{(p_s)_a}{p_a}
=\frac{e^{-\eta s v_a}}{\sum_b p_be^{-\eta s v_b}}
\le e^{\eta sR},\\*
e^{-\eta R}\Var_p(v)
&\le\Var_{p_s}(v)\le e^{\eta R}\Var_p(v).
\end{align*}
The first line uses $\osc(v)\le R$. For the second, multiply the coordinatewise probability bounds by $(v_a-c)^2$, sum, and minimize over $c$, using $s\le1$. Substituting into \cref{eq:hessian-integral} and using $\int_0^1(1-s)\,\dd s=1/2$ proves the claim.
\end{proof}

We now derive the variance-based RVU bound.

\printrepeatedstatement{lem-regret-simple}{\regretSimpleRestated}

\begin{proof}
To telescope the soft-min potential $\Phi$, distinguish cumulative loss $L^t$ from the cumulative loss plus prediction $q^t$ defining $x^t$. Let
\[
L^t:=\sum_{\tau=1}^t\ell^\tau,
\qquad
q^t:=L^{t-1}+\ell^{t-1}.
\]
Here $L^0=0$, and $q^t$ adds the prediction $\ell^{t-1}$ to past losses $L^{t-1}$. Unrolling Optimistic Hedge gives $x^t=\softmax(-\eta q^t)=\nabla\Phi(q^t)$. The divergence difference separates the incurred loss $\ip{x^t}{\ell^t}$ from the change in $\Phi$:
\[
D_\Phi(L^t,q^t)-D_\Phi(L^{t-1},q^t)
=\ip{x^t}{\ell^t}-\bigl[\Phi(L^t)-\Phi(L^{t-1})\bigr].
\]
Summing cancels the intermediate potential values. Using $\Phi(L^T)\le\min_aL_a^T$ and $\Phi(0)=-(\log d_i)/\eta$ then gives
\begin{equation}
\Reg(T)
\le\frac{\log d_i}{\eta}
 +\sum_{t=1}^T\left[D_\Phi(L^t,q^t)-D_\Phi(L^{t-1},q^t)\right].
\label{eq:regret-bregman}
\end{equation}

We bound $D_\Phi(L^t,q^t)$ above and $D_\Phi(L^{t-1},q^t)$ below. Their displacements are the prediction error $\ell^t-\ell^{t-1}$ and the negative previous loss $-\ell^{t-1}$.

The prediction-error displacement $L^t-q^t=\ell^t-\ell^{t-1}$ has coordinate span at most $2$, whereas $L^{t-1}-q^t=-\ell^{t-1}$ has span at most $1$. Applying \cref{lem:softmin-variance} to the positive and negative divergences in \cref{eq:regret-bregman} gives
\begin{equation}
\Reg(T)
\le\frac{\log d_i}{\eta}
 +\frac{\eta e^{2\eta}}2\sum_{t=1}^T\Var_{x^t}(\ell^t-\ell^{t-1})
 -\frac{\eta e^{-\eta}}2\sum_{t=1}^T\Var_{x^t}(\ell^{t-1}).
\label{eq:regret-exact}
\end{equation}
Choose a universal $\bar\eta>0$ such that $e^{2\eta}/2\le2/3$ and $e^{-\eta}/2\ge1/3$ for $0<\eta\le\bar\eta$. This proves \cref{lem:regret-simple}.
\end{proof}

\subsection{The weighted pairwise loss gaps}

The identity $\Var_{x_i^t}(\ell_i^{t-1})=\norm{h_i^t}_2^2$ represents stability exactly. Prediction error is more delicate: differencing the weighted gaps $h_i^t$ changes both the losses and their probability weights. We show that this additional weight change contributes only $O(\eta^2)$ times the squared weighted-gap magnitude, plus a bounded terminal correction.

\printrepeatedstatement{lem-weighted-gap}{\weightedGapRestated}

\begin{proof}
For a mixed strategy $x\in\Delta_d$ and an action-loss vector $v\in\R^d$, the pairwise variance identity is
\begin{equation}
\Var_x(v)=\sum_{a<b}x_ax_b(v_a-v_b)^2.
\label{eq:pairwise-variance}
\end{equation}
This identifies $\Var_{x_i^t}(\ell_i^{t-1})=\norm{h_i^t}_2^2$ and gives \cref{eq:S-gap}. To verify it, expand $\frac12\sum_{a,b}x_ax_b(v_a-v_b)^2$; the diagonal terms vanish.

To compare the consecutive weighted gaps $h_i^{t+1}$ and $h_i^t$, isolate their probability weights. Fix player $i$ and pair $a<b$, and write
\[
w_{ab}^t:=\sqrt{x_{i,a}^t x_{i,b}^t}.
\]
The weight $w_{ab}^t$ scales the pairwise gap $\ell_{i,a}^{t-1}-\ell_{i,b}^{t-1}$. To track its update, let $m_i^t:=2\ell_i^t-\ell_i^{t-1}$ and $Z_i^t:=\sum_cx_{i,c}^te^{-\eta m_{i,c}^t}$; these are the loss vector in the exponent and the normalizing constant in the next update. The Optimistic Hedge update gives
\begin{equation*}
\frac{w_{ab}^{t+1}}{w_{ab}^t}
=\frac{\exp\{-\eta(m_{i,a}^t+m_{i,b}^t)/2\}}{Z_i^t},
\qquad
e^{-\eta\max_c m_{i,c}^t}
\le Z_i^t
\le e^{-\eta\min_c m_{i,c}^t}.
\end{equation*}
Since $\osc(m_i^t)\le3$,
\begin{equation}
e^{-3\eta}\le\frac{w_{ab}^{t+1}}{w_{ab}^t}\le e^{3\eta},
\qquad
\left|\frac{w_{ab}^t}{w_{ab}^{t+1}}-1\right|\le C\eta
\label{eq:weight-ratio}
\end{equation}
for a universal $C$ and small universal $\eta$. Thus the relative weight change is $O(\eta)$, whose square produces the $O(\eta^2)$ correction.

Since $[h_i^t]_{ab}=w_{ab}^t(\ell_{i,a}^{t-1}-\ell_{i,b}^{t-1})$,
\begin{equation}
w_{ab}^t\bigl[(\ell_{i,a}^{t}-\ell_{i,b}^{t})-(\ell_{i,a}^{t-1}-\ell_{i,b}^{t-1})\bigr]
=[h_i^{t+1}-h_i^t]_{ab}
 +\left(\frac{w_{ab}^t}{w_{ab}^{t+1}}-1\right)[h_i^{t+1}]_{ab}.
\label{eq:moving-weight}
\end{equation}
In \cref{eq:moving-weight}, $[h_i^{t+1}-h_i^t]_{ab}$ is the weighted-gap change; the remaining term accounts for the weight ratio $w_{ab}^t/w_{ab}^{t+1}$. Apply $(a+b)^2\le\frac32a^2+3b^2$ to \cref{eq:moving-weight}, then sum over $t$ and $a<b$. Using \cref{eq:pairwise-variance,eq:weight-ratio},
\begin{align*}
\sum_{t=1}^T\Var_{x_i^t}(\ell_i^t-\ell_i^{t-1})
&=\sum_{t=1}^T\sum_{a<b}
 \left(w_{ab}^t\bigl[(\ell_{i,a}^{t}-\ell_{i,b}^{t})
 -(\ell_{i,a}^{t-1}-\ell_{i,b}^{t-1})\bigr]\right)^2\\*
&\le\frac32\sum_{t=1}^T\norm{h_i^{t+1}-h_i^t}_2^2
 +C\eta^2\sum_{t=1}^T\norm{h_i^{t+1}}_2^2\\*
&\le\frac32\sum_{t=1}^T\norm{h_i^{t+1}-h_i^t}_2^2+C\eta^2\sum_{t=1}^T\norm{h_i^t}_2^2+C\eta^2.
\end{align*}
The last inequality uses $h_i^1=0$ and
\[
\sum_{t=1}^T\norm{h_i^{t+1}}_2^2
=\sum_{t=1}^T\norm{h_i^t}_2^2+\norm{h_i^{T+1}}_2^2
\le \sum_{t=1}^T\norm{h_i^t}_2^2+\frac14.
\]
The terminal variance $\norm{h_i^{T+1}}_2^2\le1/4$ supplies the additive correction. The weight estimate controls the ratio $w_{ab}^{t+1}/w_{ab}^t$, so no uniform lower bound on the action probabilities is needed.
\end{proof}

\section{High-order finite-difference bounds}
\label{app:high-order}

\paragraph{Proof strategy.}
The target is the factorial finite-difference bound of \cref{lem:high-order-gap}. We first reduce the derivative norm of $F$ to derivatives of expected loss gaps, then bound their joint-cumulant representations to prove \cref{lem:F-derivative-bounds}. The finite-difference chain rule and scalar comparison then prove the centered-logit difference bound in \cref{lem:high-order-logit} from the centered-logit recurrence. Finally, \cref{lem:difference-chain-rule,lem:high-order-logit,lem:H-derivative-bounds} transfer that bound through $H_i$ to prove \cref{lem:high-order-gap}.

\paragraph{Notation.}
We use moment norms for random quantities and operator norms for derivatives. For $p\ge1$, the moment norm is $\norm{Y}_{L^p}:=(\E|Y|^p)^{1/p}$. To track all temporal orders together, we use formal power series. For a formal power series $P(s)=\sum_{k\ge0}p_ks^k$, write $[s^k]P:=p_k$, with the same notation for other indeterminates. Write $\R[[s]]$ for real formal power series and $U\preceq V$ for coefficientwise domination. This bounds each order separately without requiring series convergence.

\subsection{Fr\'echet derivatives and operator norms}
\label{app:frechet-derivatives}

The derivatives below measure how the centered loss map $F$ and weighted-gap map $H_i$ respond to perturbations of their centered-logit inputs. We first define the derivatives, then the operator norms used to bound them.

\begin{definition}[Fr\'echet derivative]
\label{def:frechet-derivative}
Let $X,Y$ be finite-dimensional normed spaces, let $U\subseteq X$ be open, and let $\Psi:U\to Y$. The Fr\'echet derivative at $x\in U$ is the unique continuous linear map $D\Psi(x):X\to Y$, when it exists, satisfying
\[
\lim_{\substack{h\to0\\h\ne0}}
\frac{\norm{\Psi(x+h)-\Psi(x)-D\Psi(x)[h]}_Y}{\norm h_X}=0.
\]
\end{definition}

\begin{definition}[Higher-order Fr\'echet derivatives]
\label{def:higher-frechet-derivative}
Let $X,Y$ be finite-dimensional normed spaces, let $U\subseteq X$ be open, and let $\Psi:U\to Y$. For an integer $r\ge2$, suppose $D^{r-1}\Psi$ is defined near $x\in U$ and is Fr\'echet differentiable at $x$ as a map into the continuous $(r-1)$-linear maps, equipped with the operator norm defined below. The $r$th Fr\'echet derivative $D^r\Psi(x):X^r\to Y$ is defined by
\[
D^r\Psi(x)[v^1,\ldots,v^r]
:=\bigl(D(D^{r-1}\Psi)(x)[v^r]\bigr)[v^1,\ldots,v^{r-1}].
\]
We set $D^0\Psi(x):=\Psi(x)$ and call $\Psi$ smooth when these derivatives exist and are continuous at every order.
\end{definition}

For smooth $\Psi$, the derivative $D^r\Psi(x)$ is symmetric. Its evaluation has a direct interpretation: perturb $x$ independently in the directions $v^1,\ldots,v^r$ and differentiate once in each perturbation amplitude:
\[
D^r\Psi(x)[v^1,\ldots,v^r]
=\left.\partial_{s_1}\cdots\partial_{s_r}
\Psi\!\left(x+\sum_{q=1}^rs_qv^q\right)\right|_{s_1=\cdots=s_r=0}.
\]
The operator norm makes this response uniform over directions of norm at most one.

\begin{definition}[Multilinear operator norm]
\label{def:operator-norm}
Let $X,Y$ be finite-dimensional normed spaces and $r\ge1$. The induced operator norm of a continuous $r$-linear map $L:X^r\to Y$ is
\[
\norm L:=\sup_{\substack{\norm{v^q}_X\le1\\q=1,\ldots,r}}
\norm{L[v^1,\ldots,v^r]}_Y.
\]
\end{definition}

\paragraph{Notation.}
We use this operator norm for $L=D^r\Psi(x)$, and set $\norm{D^0\Psi(x)}:=\norm{\Psi(x)}_Y$ at order zero. For $F:E\to E$, both spaces carry the centered-logit product norm $\norm{\cdot}_E$, and we write $\norm{D^rF(z)}_E$ for the resulting operator norm. For $H_i:E\times E\to\R^{\binom{d_i}{2}}$, the input norm is $\norm{(p,q)}=(\norm p_E^2+\norm q_E^2)^{1/2}$ and the output norm is Euclidean, as in \cref{lem:H-derivative-bounds}.

\subsection{Proof of finite-difference interpolation}
\label{app:fourier-proof}

\printrepeatedstatement{lem:fourier-tools}{\directFourierRestated}

\begin{proof}
We compare the sequence norms in an orthonormal basis of the finite-dimensional Hilbert space. Define the Fourier transform $\widehat y(\omega):=\sum_{t\in\mathbb Z}y^t e^{-it\omega}$ for $\omega\in[-\pi,\pi]$, using the complex Euclidean norm on these coordinates. The sum is finite, and shifting the time index gives $\widehat{\Delta^ry}(\omega)=(e^{i\omega}-1)^r\widehat y(\omega)$. Parseval's identity therefore yields, for every integer $r\ge0$,
\[
\norm{\Delta^ry}_{\ell_2}^2
=\frac1{2\pi}\int_{-\pi}^{\pi}
|e^{i\omega}-1|^{2r}\norm{\widehat y(\omega)}^2\,\dd\omega.
\]
For $\lambda:=|e^{i\omega}-1|^2$, the inequality $\lambda\le\theta+\theta^{1-m}\lambda^m$ holds at every frequency: the first term suffices when $\lambda\le\theta$, and the second when $\lambda>\theta$. Multiplying by $\norm{\widehat y(\omega)}^2/(2\pi)$ and integrating gives \cref{eq:direct-fourier}.
\end{proof}

\subsection{A joint-cumulant estimate}
\label{app:cumulants}

To prove \cref{lem:F-derivative-bounds}, we express derivatives of expectations under softmax distributions as joint cumulants. We first define these quantities and prove their derivative representation. We then bound them using moments, separating growth in the derivative order from the number of players and actions.

\begin{definition}[Joint cumulant]
\label{def:joint-cumulant}
Let $m\ge1$, and let $Y_1,\ldots,Y_m$ be real random variables on a common probability space with $\E\prod_{j\in B}|Y_j|<\infty$ for every nonempty $B\subseteq[m]$. Their joint cumulant is the real number
\begin{equation}
\kappa(Y_1,\ldots,Y_m)
:=\sum_{\pi\in\Pi_m}(-1)^{|\pi|-1}(|\pi|-1)!
 \prod_{B\in\pi}\E\prod_{j\in B}Y_j,
\label{eq:cumulant-definition}
\end{equation}
where $\Pi_m$ is the set of partitions of $[m]$ into nonempty blocks.
\end{definition}

For one variable, $\kappa(Y_1)=\E Y_1$; for two, $\kappa(Y_1,Y_2)=\operatorname{Cov}(Y_1,Y_2)$. The next lemma identifies every joint cumulant with a mixed derivative of the logarithm of the joint moment-generating function. Its exponential-moment assumption holds automatically for the finite-support random variables in our softmax calculations.

\begin{lemma}[Derivative representation of joint cumulants]
\label{lem:cumulant-derivative}
Let the joint moment-generating function $M(s):=\E\exp\{\sum_{j=1}^m s_jY_j\}$ be finite in a neighborhood of $s=0\in\R^m$. Then
\[
\kappa(Y_1,\ldots,Y_m)
=\left.\partial_{s_1}\cdots\partial_{s_m}\log M(s)\right|_{s=0}.
\]
\end{lemma}

\begin{proof}
The exponential-moment assumption permits Taylor expansion and differentiation under the expectation near zero. Let $[s_1\cdots s_m]$ denote the coefficient of the monomial in which each variable appears once. Since $M(0)=1$, only the first $m$ terms of $\log M=\sum_{k\ge1}(-1)^{k-1}(M-1)^k/k$ contribute to this coefficient. Each contribution assigns a nonempty subset of $[m]$ to each of the $k$ factors. Thus
\begin{align*}
\left.\partial_{s_1}\cdots\partial_{s_m}\log M(s)\right|_{s=0}
&=[s_1\cdots s_m]\log M(s)\\
&=\sum_{k=1}^m\frac{(-1)^{k-1}}{k}
 \sum_{(B_1,\ldots,B_k)}\prod_{q=1}^k\E\prod_{j\in B_q}Y_j\\
&=\sum_{\pi\in\Pi_m}(-1)^{|\pi|-1}(|\pi|-1)!
 \prod_{B\in\pi}\E\prod_{j\in B}Y_j.
\end{align*}
The inner sum is over ordered partitions of $[m]$ into $k$ nonempty blocks. Each unordered partition has $k!$ orderings, giving the factor $k!/k=(k-1)!$ in the last line. This is \cref{eq:cumulant-definition}.
\end{proof}

The following lemma bounds a joint cumulant by the product of its arguments' moment scales, with a single factorial in the order. It does not require independence among the arguments.

\begin{lemma}[Joint-cumulant bound]
\label{lem:cumulant}
Let $m\ge2$, and let $Y_1,\ldots,Y_m$ be real random variables on the same probability space. Suppose that numbers $K_1,\ldots,K_m\ge0$ satisfy
\[
\norm{Y_j}_{L^p}\le K_j\sqrt p
\qquad\text{for every }p\ge1\text{ and }j\in[m].
\]
Then
\begin{equation}
\abs{\kappa(Y_1,\ldots,Y_m)}
\le(2e)^m(m-1)!\prod_{j=1}^mK_j.
\label{eq:cumulant-bound}
\end{equation}
Moreover, a joint cumulant of order $m\ge2$ is unchanged when any one of its arguments is shifted by a constant.
\end{lemma}

\begin{proof}
We bound the moment for each block of a partition, then evaluate the weighted partition sum exactly. For a block $B$ of size $b$, generalized H\"older gives
\[
\abs{\E\prod_{j\in B}Y_j}
\le\prod_{j\in B}\norm{Y_j}_{L^b}
\le b^{b/2}\prod_{j\in B}K_j
\le e^b b!\prod_{j\in B}K_j.
\]
The last step uses $b!\ge(b/e)^b$. For $k$ ordered blocks, there are $\binom{m-1}{k-1}$ positive tuples of block sizes summing to $m$. For each tuple, the multinomial count $m!/(b_1!\cdots b_k!)$ cancels the product of block factorials. Removing the $k!$ block orderings and multiplying by the cumulant weight $(k-1)!$ leaves the factor $1/k$. Hence
\begin{align*}
\sum_{\pi\in\Pi_m}(|\pi|-1)!\prod_{B\in\pi}|B|!
&=m!\sum_{k=1}^m\frac1k\binom{m-1}{k-1}\\*
&=(m-1)!\sum_{k=1}^m\binom{m}{k}
=(m-1)!(2^m-1).
\end{align*}
Consequently,
\[
\abs{\kappa(Y_1,\ldots,Y_m)}
\le e^m(m-1)!(2^m-1)\prod_{j=1}^mK_j
\le(2e)^m(m-1)!\prod_{j=1}^mK_j.
\]
For shift invariance, \cref{eq:cumulant-definition} is multilinear, so it suffices to show that a cumulant with one constant argument is zero. Fix a partition of the other $m-1$ indices into $k$ blocks. Making the constant a singleton has coefficient $(-1)^kk!$; adjoining it to one of the $k$ blocks has total coefficient $k(-1)^{k-1}(k-1)!$. These contributions have the same moment product and cancel.
\end{proof}

\subsection{Proof of the centered-loss derivative bound}
\label{app:loss-derivative-proof}

For the derivative bound below, the cumulant arguments are a bounded loss gap $g$ and centered opponent sums $S_1,\ldots,S_r$, defined in the proof. We apply \cref{lem:cumulant} with $m=r+1$, $Y_1=g$, and $Y_{q+1}=S_q$ for $q\in[r]$. The moment bounds $\norm g_{L^p}\le1$ and $\norm{S_q}_{L^p}\le C\sqrt p\,\norm{v^q}_E$ give $K_1=1$ and $K_{q+1}=C\norm{v^q}_E$, yielding the required $r!$ derivative bound without independence among the $S_q$.

\printrepeatedstatement{lem-F-derivative-bounds}{\lossResponseDerivativeRestated}

\begin{proof}
The proof strategy is to bound each scalar derivative $D^r(F_{i,a}-F_{i,b})$ before combining the player blocks. Fix $r\ge1$, a centered-logit state $z\in E$, and perturbation directions $v^1,\ldots,v^r\in E$. Write $v^q=(v_1^q,\ldots,v_n^q)$, where $v_j^q\in E_j$ is the player-$j$ block of direction $q$, and $v_j^q(c)$ is its coordinate for action $c\in A_j$. Write $F_{i,a}$ for action $a$'s coordinate of $F_i$. The product norm gives
\begin{align*}
\norm{D^rF(z)[v^1,\ldots,v^r]}_E
&=\left(\sum_{i=1}^n\osc\!\left(D^rF_i(z)[v^1,\ldots,v^r]\right)^2\right)^{1/2}\\*
&\le\sqrt n\,\max_i\max_{a,b\in A_i}
 \abs{D^r(F_{i,a}-F_{i,b})(z)[v^1,\ldots,v^r]}.
\end{align*}
The inequality uses that a vector's oscillation is its largest absolute coordinate difference. It therefore suffices to bound every scalar derivative in the last line by $AR^rr!\prod_{q=1}^r\norm{v^q}_E$, uniformly over $i,a,b$.

Fix such a player $i$ and actions $a,b\in A_i$. Introduce independent opponent actions $\mathsf A_j\sim\softmax(z_j)$ for $j\ne i$, and write $\mathsf A_{-i}:=(\mathsf A_j)_{j\ne i}$. Define the loss gap $g$ and centered directional sums $S_q$ by
\[
g:=\ell_i(b,\mathsf A_{-i})-\ell_i(a,\mathsf A_{-i}),
\qquad
S_q:=\sum_{j\ne i}\bigl(v_j^q(\mathsf A_j)-\E v_j^q(\mathsf A_j)\bigr),
\qquad q\in[r].
\]
The random variable $g$ takes values in $[-1,1]$. Each $S_q$ is the directional derivative at $z$ of the opponents' joint log-probability along $v^q$, evaluated at their random action profile. All expectations and centering means below use this fixed law at $z$. Since $P_i$ subtracts the same value from every coordinate, $(F_{i,a}-F_{i,b})(z)=\E g$.

We next express derivatives of this expectation as cumulants. Let $s_1,\ldots,s_r\in\R$ specify the perturbation amplitudes. Softmax multiplies the opponents' joint probability weights by $\exp\{\sum_qs_q\sum_{j\ne i}v_j^q(\mathsf A_j)\}$ and then normalizes. Subtracting the means in $S_q$ changes every weight by a common factor that cancels on normalization. Introducing an auxiliary scalar $u\in\R$ to extract the loss gap by differentiation gives
\begin{align*}
(F_{i,a}-F_{i,b})\left(z+\sum_{q=1}^rs_qv^q\right)
&=\frac{\E\!\left[g\exp\{\sum_{q=1}^rs_qS_q\}\right]}
        {\E\exp\{\sum_{q=1}^rs_qS_q\}}\\*
&=\left.\partial_u\log\E\exp\left\{ug+\sum_{q=1}^rs_qS_q\right\}\right|_{u=0}.
\end{align*}
The auxiliary scalar $u$ inserts the loss gap $g$ through differentiation, while the logarithm supplies the normalizing denominator. Differentiating once in each $s_q$ and setting the perturbation amplitudes to zero gives, by \cref{lem:cumulant-derivative},
\begin{equation}
D^r(F_{i,a}-F_{i,b})(z)[v^1,\ldots,v^r]
=\kappa(g,S_1,\ldots,S_r).
\label{eq:loss-derivative-cumulant}
\end{equation}
All arguments have finite support, which justifies the differentiations and the cumulant representation.

It remains to bound this cumulant. Since $\abs g\le1$, \cref{lem:cumulant} reduces the task to controlling the moments of each $S_q$. Its summands are independent across opponents, centered, and have range lengths $\osc(v_j^q)$. Hoeffding's lemma and the resulting subgaussian moment bound give
\begin{align*}
\E e^{\lambda S_q}
&\le\exp\left\{\frac{\lambda^2}{8}\sum_{j\ne i}\osc(v_j^q)^2\right\}
\le\exp\left\{\frac{\lambda^2}{8}\norm{v^q}_E^2\right\},
\qquad \lambda\in\R,\\*
\norm{S_q}_{L^p}
&\le C\sqrt p\,\norm{v^q}_E,
\qquad p\ge1.
\end{align*}
The second line follows by integrating the subgaussian tail implied by the first; $C$ is universal. The sums $S_q$ may depend on one another, but \cref{lem:cumulant} requires only their individual moment bounds. Applying it to the $r+1$ arguments in \cref{eq:loss-derivative-cumulant}, using $\norm g_{L^p}\le1$, yields
\begin{align*}
\abs{D^r(F_{i,a}-F_{i,b})(z)[v^1,\ldots,v^r]}
&\le (2e)^{r+1}C^rr!\prod_{q=1}^r\norm{v^q}_E\\
&\le AR^rr!\prod_{q=1}^r\norm{v^q}_E
\end{align*}
for universal $A,R>0$. Maximizing over $a,b$ proves \cref{eq:Fi-derivative}; substituting into the opening norm inequality proves \cref{eq:F-derivative}. All constants are uniform over $z,i,a,b$, and the only factor $\sqrt n$ comes from combining the player blocks.
\end{proof}

\subsection{Proof of the finite-difference chain rule}

We now convert bounds on the derivatives $D^p\Psi$ and input differences $\Delta^ru^t$ into bounds on the output differences $\Delta^k\Psi(u^t)$. The key issue is combinatorial: a separate bound on the number of partitions would introduce an additional factorial, preventing the difference bound used in the interpolation argument.

\printrepeatedstatement{lem-difference-chain-rule}{\differenceChainRuleRestated}

\begin{proof}
The proof strategy is to express $\Delta^k\Psi(u^t)$ as an integral of a mixed derivative over $[0,1]^k$. Multilinear interpolation makes the inner derivatives averages of input differences, and counting the resulting partitions gives the coefficient formula in \cref{eq:difference-chain-rule}. Fix $t\ge k+1$ and define the multilinear interpolation $U:[0,1]^k\to X$ by
\[
U(\theta):=\sum_{\varepsilon\in\{0,1\}^k}
 u^{t+|\varepsilon|}
 \prod_{j=1}^k\theta_j^{\varepsilon_j}(1-\theta_j)^{1-\varepsilon_j}.
\]
This interpolant agrees with the trajectory values at the cube's vertices: $U(\varepsilon)=u^{t+|\varepsilon|}$. Repeated use of the fundamental theorem of calculus gives
\[
\Delta^k\Psi(u^t)
=\int_{[0,1]^k}\partial_1\cdots\partial_k(\Psi\circ U)(\theta)\,\dd\theta.
\]
To track the input derivative associated with a block, write $\partial_B:=\prod_{j\in B}\partial_j$ for nonempty $B\subseteq[k]$. Differentiating $U$ shows that $\partial_BU(\theta)$ is a convex combination of vectors $\Delta^{|B|}u^{t+j}$ with $0\le j\le k-|B|$. Since $t\ge k+1\ge|B|+1$, all these indices lie in the range defining $a_{|B|}$, and hence
\[
\norm{\partial_BU(\theta)}\le |B|!a_{|B|}.
\]
The mixed Fa\`a di Bruno formula is indexed by the partitions of $[k]$:
\[
\partial_1\cdots\partial_k(\Psi\circ U)
=\sum_{\pi\in\Pi_k}D^{|\pi|}\Psi(U)[\partial_BU:B\in\pi].
\]
Each block supplies an input derivative, while the number of blocks sets the outer derivative order. Consequently,
\[
\frac{\norm{\Delta^k\Psi(u^t)}}{k!}
\le\frac1{k!}\sum_{p=1}^kp!c_p
 \sum_{\substack{\pi\in\Pi_k\\|\pi|=p}}
 \prod_{B\in\pi}|B|!a_{|B|}.
\]
Multiplying the sum over unordered $p$-block partitions by $p!$ is the same as summing over ordered blocks. For a fixed ordered size tuple $(r_1,\ldots,r_p)$ with every $r_j\ge1$ and $\sum_jr_j=k$, there are exactly $k!/(r_1!\cdots r_p!)$ ordered partitions of those sizes. After division by $k!$, their contribution is $c_p\prod_{j=1}^pa_{r_j}$. Summing over the ordered size tuples gives exactly the coefficient on the right-hand side of Equation~\eqref{eq:difference-chain-rule} in \Cref{lem:difference-chain-rule}. Taking the supremum over $t\ge k+1$ proves the claim.
\end{proof}

\subsection{Proof of the centered-logit difference bound}
\label{app:logit-difference-proof}

To prove \cref{lem:high-order-logit}, we divide each order-$k$ centered-logit difference by $k!$ and bound the resulting coefficients geometrically. The following scalar comparison will control the recursion obtained from the chain rule and centered-logit recurrence.

\begin{lemma}[Scalar coefficient comparison]
\label{lem:formal-comparison}
Let $K,R>0$, and let $A(s)=\sum_{k\ge1}a_ks^k$ be a formal power series with nonnegative coefficients and $A(0)=0$. If $A'(s)\preceq K/(1-RA(s))$, then $a_k\le C_R(C_RK)^k$ for a constant $C_R\ge1$ depending only on $R$.
\end{lemma}

\begin{proof}
The proof strategy is coefficientwise comparison with a scalar equality. Since each new coefficient depends only on earlier coefficients, induction suffices without assuming convergence of the formal series. Let $W(s)=\sum_{k\ge1}w_ks^k$ be the unique formal solution of
\[
W'(s)=\frac{K}{1-RW(s)},
\qquad W(0)=0.
\]
The coefficient recursion of this equality matches the coefficientwise differential inequality in \cref{lem:formal-comparison} with equality, making $W$ a candidate upper bound for $A$. Solving it gives
\[
W(s)=\frac{1-\sqrt{1-2KRs}}{R}.
\]
We prove $A\preceq W$ coefficientwise. Suppose inductively that $a_j\le w_j$ for $1\le j<k$. The coefficient of degree $k-1$ in
\[
(1-RA(s))^{-1}=\sum_{p\ge0}(RA(s))^p
\]
is a polynomial with nonnegative coefficients in $a_1,\ldots,a_{k-1}$; it does not involve $a_k$ because $A(0)=0$. The same polynomial evaluated at $w_1,\ldots,w_{k-1}$ is the degree-$(k-1)$ coefficient of $(1-RW)^{-1}$. Therefore
\[
k a_k=[s^{k-1}]A'(s)
\le K[s^{k-1}](1-RA)^{-1}
\le K[s^{k-1}](1-RW)^{-1}
=k w_k.
\]
The induction starts at $k=1$, so $A\preceq W$. Finally, the binomial expansion of $\sqrt{1-2KRs}$ gives a constant $C_R\ge1$, depending only on $R$, such that $w_k\le C_R(C_RK)^k$. This proves the lemma.
\end{proof}

We now apply the scalar comparison to the centered-logit differences.

\printrepeatedstatement{lem-high-order-logit}{\highOrderLogitRestated}

\begin{proof}
The proof strategy is to combine the derivative bounds for $F$, the finite-difference chain rule, and the centered-logit recurrence into the coefficient inequality of \cref{lem:formal-comparison}. Define the factorially normalized centered-logit and loss differences, and their formal series, by
\begin{align*}
a_k&:=\sup_{t\ge k+1}\frac{\norm{\Delta^kz^t}_E}{k!},
&b_k&:=\sup_{t\ge k+1}\frac{\norm{\Delta^kF(z^t)}_E}{k!},\\
A(s)&:=\sum_{k\ge1}a_ks^k,
&B(s)&:=\sum_{k\ge1}b_ks^k.
\end{align*}
These coefficients are finite before any high-order estimate is used: the centered-logit recurrence gives $\norm{\Delta z^t}_E\le3\sqrt n\,\eta$, so $a_k\le3\sqrt n\,\eta\,2^{k-1}/k!$, while $\norm{F(z)}_E\le\sqrt n$ gives $b_k\le2^k\sqrt n/k!$. The series track factorially normalized differences, have nonnegative coefficients, and satisfy $A(0)=B(0)=0$. By \cref{lem:F-derivative-bounds,lem:difference-chain-rule}, for universal $A_0,R>0$,
\begin{equation}
B(s)\preceq A_0\sqrt n\,\frac{RA(s)}{1-RA(s)}.
\label{eq:F-difference-series}
\end{equation}
Indeed, for every $k\ge1$,
\[
b_k\le [s^k]\sum_{p=1}^k A_0\sqrt n\,R^pA(s)^p
\le [s^k]A_0\sqrt n\sum_{p\ge1}(RA(s))^p.
\]
Since $\norm{F(z)}_E\le\sqrt n$, the centered-logit recurrence gives $a_1\le3\sqrt n\,\eta$. For $k\ge2$, applying $\Delta^{k-1}$ to the increment $\Delta z^t$ gives
\[
\Delta^kz^t
=\eta\left(2\Delta^{k-1}F(z^t)-\Delta^{k-1}F(z^{t-1})\right).
\]
For $t\ge k+1$, both starting indices $t$ and $t-1$ lie in the range $s\ge k$ defining $b_{k-1}$. Dividing by $(k-1)!$ and taking the supremum therefore gives
\begin{equation}
ka_k\le3\eta b_{k-1},
\qquad k\ge2.
\label{eq:ak-bk-recurrence}
\end{equation}
Let $M:=\max\{1,A_0\}$. Coefficientwise,
\[
1+A_0\frac{RA}{1-RA}
=1+A_0\sum_{p\ge1}(RA)^p
\preceq M\sum_{p\ge0}(RA)^p
=\frac{M}{1-RA}.
\]
Combining this bound with \cref{eq:F-difference-series,eq:ak-bk-recurrence} gives
\begin{align}
A'(s)
&=a_1+\sum_{k\ge2}ka_ks^{k-1}
 \preceq3\sqrt n\,\eta+3\eta B(s)\notag\\*
&\preceq3\sqrt n\,\eta
 \left(1+A_0\frac{RA(s)}{1-RA(s)}\right)
\label{eq:A-prime-before-comparison}\\*
&\preceq\frac{3M\sqrt n\,\eta}{1-RA(s)}.
\label{eq:logit-series-main}
\end{align}
Thus \cref{eq:logit-series-main} has the form required by \cref{lem:formal-comparison}:
\begin{equation}
A'(s)\preceq\frac{K}{1-RA(s)},
\qquad K:=3M\sqrt n\,\eta.
\label{eq:A-prime-coefficient-bound}
\end{equation}
The series $A$ has nonnegative coefficients, $A(0)=0$, and $1-RA(s)$ has constant coefficient one. Since $R$ and $M$ are universal, \cref{lem:formal-comparison} gives $a_k\le C_z(R_z\sqrt n\,\eta)^k$ for universal constants $C_z,R_z>0$. Multiplying by $k!$ proves \cref{eq:z-high-order} and completes the proof of \cref{lem:high-order-logit}.
\end{proof}

\subsection{Proof of the weighted-gap derivative bound}

The output $H_i(z,z')$ has $\binom{d_i}{2}$ coordinates. To retain logarithmic dependence on $d_i$ in \cref{eq:log-regret}, we sum squared derivatives using probability normalization rather than count these pairs.

\printrepeatedstatement{lem-H-derivative-bounds}{\weightedGapDerivativeRestated}

\begin{proof}
We first show that each derivative of a square-root probability is bounded by that square-root probability times a factorial bound. The trilinear map in \cref{eq:B-operator} then combines those factors with the loss gaps in $F_i$. Define the square-root probability map $\rho_i:E_i\to\R^{d_i}$ by $\rho_i(z_i):=\sqrt{\softmax(z_i)}$ coordinatewise, and write $\rho_{i,a}=e^{\psi_a}$, where the logarithm $\psi_a$ is
\[
\psi_a(z_i)=\frac12z_{i,a}-\frac12\log\sum_b e^{z_{i,b}}.
\]
The logarithm separates the action coordinate from the common normalizer, whose derivatives are cumulants. For every direction $p$,
\[
D\psi_a(z_i)[p]=\frac12\left(p_a-\E_{\mathsf A\sim\softmax(z_i)}p_{\mathsf A}\right).
\]
For $r\ge2$, $D^r\psi_a[p^1,\ldots,p^r]$ is $-1/2$ times the joint cumulant of $p_{\mathsf A}^1,\ldots,p_{\mathsf A}^r$. Centering the arguments, applying Hoeffding's lemma, and then \cref{lem:cumulant} give universal constants $C_\psi,R_\psi$ such that
\[
|D^r\psi_a(z_i)[p^1,\ldots,p^r]|
\le C_\psi R_\psi^r r!\prod_{q=1}^r\osc(p^q),
\qquad r\ge1.
\]
For $r\ge1$, the exponential Fa\`a di Bruno formula gives
\[
D^r\rho_{i,a}(z_i)[p^1,\ldots,p^r]
=\rho_{i,a}(z_i)\sum_{\pi\in\Pi_r}
 \prod_{B\in\pi}D^{|B|}\psi_a(z_i)[p^q:q\in B].
\]
Every term retains the same prefactor $\rho_{i,a}(z_i)$. Combining the derivative bound for $\psi_a$ with
\[
\sum_{\pi\in\Pi_r}\prod_{B\in\pi}|B|!\le r!2^{r-1}
\]
gives a universal constant $C_\rho>0$ such that
\begin{equation}
|D^r\rho_{i,a}(z_i)[p^1,\ldots,p^r]|
\le\rho_{i,a}(z_i)C_\rho^r r!\prod_{q=1}^r\osc(p^q),
\qquad r\ge0.
\label{eq:rho-derivative}
\end{equation}
The case $r=0$ is immediate. Since $\sum_a\rho_{i,a}^2=1$, summing the squared coordinate bounds introduces no factor depending on $d_i$:
\begin{equation}
\norm{(D^r\rho_{i,a}[p^1,\ldots,p^r])_a}_2
\le C_\rho^r r!\prod_{q=1}^r\osc(p^q).
\label{eq:rho-vector-derivative}
\end{equation}

To control all pairs together, separate the two probability factors from the loss gap. Define the trilinear map $B_i:(\R^{d_i})^3\to\R^{\binom{d_i}{2}}$ by
\[
[B_i(r,s,w)]_{ab}:=r_as_b(w_b-w_a),
\qquad a<b.
\]
The map $B_i$ forms weighted coordinate differences and satisfies a bound independent of $d_i$:
\begin{equation}
\norm{B_i(r,s,w)}_2
\le\norm r_2\norm s_2\osc(w),
\label{eq:B-operator}
\end{equation}
because the sum over $a<b$ is bounded by the sum over all ordered pairs. Moreover,
\[
H_i(z,z')=B_i\bigl(\rho_i(z_i),\rho_i(z_i),F_i(z')\bigr).
\]
Since $\osc(F_i(z'))\le1$, the playerwise centered-loss derivative bound also holds at order zero after enlarging $A$. For labeled directions $\xi^\ell=(p^\ell,q^\ell)$, apply the multilinear Leibniz rule and group the directions assigned to the two $\rho_i$ factors and the $F_i$ factor. Combining Equations~\eqref{eq:rho-vector-derivative}, \eqref{eq:B-operator}, and \eqref{eq:Fi-derivative} gives
\begin{align*}
\norm{D^rH_i(z,z')[\xi^1,\ldots,\xi^r]}_2
&\le\sum_{p+j+k=r}\frac{r!}{p!j!k!}
 (C_\rho^pp!)(C_\rho^jj!)(AR^kk!)
 \prod_{\ell=1}^r\norm{\xi^\ell}\\
&\le A(2C_\rho+R)^r r!
 \prod_{\ell=1}^r\norm{\xi^\ell}.
\end{align*}
This proves Equation~\eqref{eq:H-derivative} in \Cref{lem:H-derivative-bounds}. The identity $\norm{\rho_i}_2=1$ and the bound in \cref{eq:B-operator} make the resulting constants independent of $d_i$.
\end{proof}

The proof of \cref{lem:high-order-gap} is given in \cref{sec:high-order-proof}.

\section{Finite-support cutoff construction}
\label{app:finite-horizon}

We prove \cref{lem:boundary-extension} by multiplying the forward trajectory $(u^t)_{t\ge1}$ by a cutoff and setting the result to zero outside a finite interval. The sequence $y$ is defined on all integer times, but need not equal $u^t$ at every $1\le t\le T$: the cutoff tapers the endpoints and can use the continuation beyond $T+1$. Its squared-norm errors are controlled by the lemma, so finite-difference interpolation (\cref{lem:fourier-tools}) can be applied.

\subsection{Proof of the finite-support cutoff lemma}

\printrepeatedstatement{lem:boundary-extension}{\finiteSupportCutoffRestated}

\begin{proof}
The proof strategy is to construct a cutoff whose $j$th differences are bounded by $(2/L)^j$ for $j\le m$, with a fixed averaging length $L\ge2$. Repeated averaging gives these bounds and a transition over $O(Lm)$ indices, while placing the cutoff after initialization avoids indices where the high-order assumption does not apply. We implement this construction by convolution and bound the $\ell_1$ norms of the resulting kernel differences. For finitely supported scalar sequences $f,g$ on $\mathbb Z$, define
\[
(f*g)^t:=\sum_{s\in\mathbb Z}f^{t-s}g^s,
\qquad
\norm{f}_{\ell_1}:=\sum_{t\in\mathbb Z}|f^t|.
\]
Convolution averages shifted copies when one factor has nonnegative entries summing to one; the $\ell_1$ norm controls the size of kernel differences.

We begin with uniform averaging over an interval. For $S\subseteq\mathbb Z$, let $\mathbf 1_S$ denote its indicator sequence, and define
\[
b^t:=\frac1L\mathbf 1_{\{0,\ldots,L-1\}}(t).
\]
The sequence $b$ averages over $L$ consecutive indices. Let $\nu=b^{*m}$ be its $m$-fold convolution, with total mass one and support in $\{0,\ldots,m(L-1)\}$. Since $\Delta(f*g)=(\Delta f)*g=f*(\Delta g)$, we can assign one difference to each of $j$ distinct factors; the other factors retain unit $\ell_1$ norm. Thus
\begin{equation}
\norm{\Delta^j\nu}_{\ell_1}
\le \norm{\Delta b}_{\ell_1}^j\norm b_{\ell_1}^{m-j}
=\left(\frac2L\right)^j,
\qquad 0\le j\le m.
\label{eq:boundary-kernel}
\end{equation}
The inequality uses $\norm{f*g}_{\ell_1}\le\norm f_{\ell_1}\norm g_{\ell_1}$. Each additional difference multiplies the bound by $2/L$, with no combinatorial coefficient; the support interval of $\nu$ has length $m(L-1)+1$.

First consider a horizon satisfying
\begin{equation}
T+1\ge2m+2+m(L-1).
\label{eq:long-horizon-high-order}
\end{equation}
This leaves an interval where the cutoff equals one. Use the given values of $u^t$ for $t\ge1$ and extend arbitrarily to nonpositive indices; those values will be irrelevant.

We smooth the interval indicator after initialization and multiply the trajectory by the resulting cutoff. Define
\[
\begin{aligned}
\chi&:=\mathbf 1_{\{2m+2,\ldots,T+1\}}*\nu,\\
y^t&:=
\begin{cases}
\chi^tu^t,&2m+2\le t\le T+1+m(L-1),\\
0,&\text{otherwise}.
\end{cases}
\end{aligned}
\]
The sequence $y$ agrees with $u$ wherever $\chi=1$. In particular, $0\le\chi^t\le1$, $y$ is finitely supported, and
\begin{align}
\chi^t&=0,
\qquad t\notin\{2m+2,\ldots,T+1+m(L-1)\},
\label{eq:boundary-zero-range}\\
\chi^t&=1,
\qquad 2m+2+m(L-1)\le t\le T+1,
\label{eq:boundary-central-range}\\
\sup_t\abs{\Delta^j\chi^t}
&\le\norm{\Delta^j\nu}_{\ell_1}
\le(2/L)^j,
\qquad0\le j\le m.
\label{eq:boundary-weight-differences}
\end{align}

These properties identify the interval where $y^t=u^t$ and control $\Delta^j\chi^t$. To bound $\Delta^my^t$, split differences between the cutoff $\chi^t$ and the trajectory $u^t$:
\begin{equation}
\Delta^m y^t
=\sum_{j=0}^m\binom mj
 (\Delta^j\chi)^{t+m-j}\Delta^{m-j}u^t.
\label{eq:product-rule-high-order}
\end{equation}
Each term assigns $j$ differences to the cutoff and $m-j$ to the trajectory. We must check that each nonzero term uses $\Delta^ru^t$ only at indices $t\ge r+2$, where the assumed bound applies. Fix a nonzero summand with $j<m$ and set $r=m-j\ge1$. The factor $(\Delta^j\chi)^{t+r}$ depends on indices $t+r,\ldots,t+m$, at least one of which must lie in the cutoff's support. Since $\chi^s=0$ for $s<2m+2$, we obtain $t\ge m+2\ge r+2$, as required by \cref{eq:boundary-extension-assumption}. For $j=m$, a nonzero $(\Delta^m\chi)^t$ likewise forces $t\ge m+2$, so $u^t$ belongs to the original sequence and is bounded by $B$.

Let $M:=\max\{B,C_0\}$. Since $r!\le m^r$ for $0\le r\le m$, Equation~\eqref{eq:boundary-extension-assumption} in \Cref{lem:boundary-extension}, together with Equations~\eqref{eq:boundary-weight-differences} and \eqref{eq:product-rule-high-order}, gives
\begin{equation}
\sup_t\norm{\Delta^my^t}
\le M\sum_{j=0}^m\binom mj(2/L)^j(R_0\rho m)^{m-j}
=M\left(\frac2L+R_0\rho m\right)^m.
\label{eq:high-order-mth-sup}
\end{equation}
The factor $2/L$ comes from the differences of $\chi$, and $R_0\rho m$ comes from the differences of $u$. Since $\Delta^my$ is nonzero at no more than $T+1+mL$ indices, enlarging a constant depending only on $B,C_0,L$ gives
\begin{equation}
\norm{\Delta^my}_{\ell_2}^2
\le C_L(T+m)\left(\frac2L+R_0\rho m\right)^{2m}.
\label{eq:high-order-mth-l2}
\end{equation}

Whenever both $t$ and $t+1$ lie in the central interval in \cref{eq:boundary-central-range}, one has $\Delta y^t=\Delta u^t$. At most $2m+m(L-1)+1=O_L(m)$ target indices $t\in\{1,\ldots,T\}$ are excluded. The additional right-boundary interval also contains $O_L(m)$ indices. Since $\norm{u^t}\le B$ and $\norm{\Delta u^t}\le2B$,
\begin{align}
\sum_{t=1}^T\norm{\Delta u^t}^2
&\le\norm{\Delta y}_{\ell_2}^2+C_Lm,
\label{eq:restore-boundary-differences}\\*
\norm{y}_{\ell_2}^2
&\le\sum_{t=1}^{T+1}\norm{u^t}^2+C_Lm.
\label{eq:restore-boundary-level}
\end{align}
The two comparisons involve at most $O_L(m)$ additional or excluded indices, giving an additive bound of $C_Lm$ in each inequality. These are the first two conclusions of \cref{lem:boundary-extension}; \cref{eq:high-order-mth-l2} supplies the third.

If the preceding horizon condition fails, then $T=O_L(m)$, so the bound $\norm{\Delta u^t}\le2B$ controls the entire difference sum. When \cref{eq:long-horizon-high-order} fails, set $y^t=0$ for all $t$. This discards all target differences, but $T\le(2+L)m+1$, so
\[
\sum_{t=1}^T\norm{\Delta u^t}^2\le4B^2T\le C_Lm,
\]
while the other two conclusions are immediate. This completes the proof.
\end{proof}

\section{Detailed Comparison with \texorpdfstring{\protect\citet{DFG2021}}{Daskalakis et al. (2021)}}
\label{app:dfg-comparison}

Both analyses reduce individual regret to high-order smoothness of self-play. We compare the estimates that permit a larger step size, using our notation $n$ for the number of players and suppressing universal numerical constants.

\subsection{The common objective}

By \cref{lem:regret-simple}, bounding the prediction-error variance sum by half the stability variance sum allows their contributions to regret to cancel. Regret is then bounded by $(\log d_i)/\eta$ plus $(2\eta/3)$ times the comparison error. The objective is a relative bound, not separate estimates of the two sums.

The central idea of \citet[Section~4]{DFG2021} is to obtain this relative comparison from the high-order smoothness of self-play, rather than from a separate bound on each round's prediction error. Their proof first bounds high-order loss differences and then returns to the first-order quantity in the RVU bound. Our proof follows the same strategy, but controls the weighted gaps $h_i^t$ through the centered-logit recurrence and compares their finite differences in one fixed norm.

\subsection{The analysis of \texorpdfstring{\protect\citet{DFG2021}}{Daskalakis et al. (2021)}}

Let $H:=\lceil\log_2T\rceil$ be the largest difference order. The scale comparison below concerns the regime $T\ge4$ and $\eta\ge1/T$ used in the absorption argument of \citet{DFG2021}; their proof treats smaller horizons and step sizes separately.

\paragraph{Obtaining a sufficiently small high-order term.}
The upward induction of \citet{DFG2021} alternates between loss differences and strategy differences, using a boundedness chain rule for softmax-type maps. Their detailed estimate has the form
\[
\norm{\Delta^k\ell_i^t}_\infty
\le (Cn\eta)^k k^{3k+1}
\le H(Cn\eta H^3)^k,
\qquad 1\le k\le H,\quad 1\le t\le T-k,
\]
for a universal $C>0$ when $\eta$ is sufficiently small relative to $1/(nH)$ \citep[Lemma~4.4, Appendix~B.4]{DFG2021}. The factor $k^{3k+1}$ matters because the estimate must remain useful as $k$ grows with $\log T$. To start and maintain the subsequent downward comparison, \citet{DFG2021} choose $Cn\eta H^3$ at most a sufficiently small multiple of $H^{-1/2}$. This gives the sufficient high-order restriction $\eta\lesssim(nH^{7/2})^{-1}$ \citep[Appendix~C.3, Equation~(96)]{DFG2021}.

\paragraph{Returning from high orders to first differences.}
For one player, write $E_k:=\sum_{t=1}^{T-k}\Var_{x_i^t}(\Delta^k\ell_i^t)$ for the weighted variance sum at order $k$, where $0\le k\le H$. \citet{DFG2021} propagate comparisons between $E_{k+1}$ and $E_k$ downward from $k=H-1$ to $k=0$, keeping the comparison coefficient of order $1/H$. The high-order estimate starts this induction; the difficult part is that the weights $x_i^t$ change with time. Even a fixed loss vector is therefore measured by different variance forms across rounds.

\citet{DFG2021} handle this by comparing the moving weights with fixed weights on short time intervals, where Fourier analysis applies. If $\beta$ denotes the adjacent-order comparison coefficient, their argument uses intervals of length $S\asymp\beta^{-3}$. Optimistic Hedge changes consecutive probability ratios by $1+O(\eta)$, so the relative change across such an interval is $O(\eta S)$ when $\eta S$ is small. Requiring this error to be of order $\beta$ gives
\[
S\asymp\beta^{-3},\qquad \eta S\lesssim\beta
\quad\Longrightarrow\quad
\eta\lesssim\beta^4\asymp H^{-4}.
\]
This is the second sufficient restriction, coming from the changing weights rather than the high-order loss estimate \citep[Lemma~4.7, Appendix~C.2]{DFG2021}. The downward induction, including the initial and terminal weight adjustments, yields the variance comparison \citep[Lemma~4.2, Appendix~C.3]{DFG2021}
\[
\sum_{t=1}^T\Var_{x_i^t}(\ell_i^t-\ell_i^{t-1})
\le\frac12\sum_{t=1}^T\Var_{x_i^t}(\ell_i^{t-1})+O(H^5)
\]
Their common choice $\eta\asymp(nH^4)^{-1}$ satisfies both restrictions. The variance error contributes $O(\eta H^5)=O(H/n)$ to regret, while the entropy term gives $O(n\log d_iH^4)$ for $d_i\ge2$; a player with one action has zero regret. These are restrictions of the proof, not lower bounds on the performance of Optimistic Hedge.

\subsection{Our approach: sharper differences in a fixed norm}

\paragraph{Controlling the trajectory through its recurrence.}
Instead of closing separate induction estimates for losses and strategies, we use the centered-logit state $z^t$ and centered loss map $F$. The finite-difference recurrence \cref{eq:logit-difference-recurrence} expresses $\Delta^kz^t$ as $\eta$ times two order-$(k-1)$ differences of $F(z^t)$. Factorial derivative bounds for $F$, combined with the normed-space chain rule, turn this relation into recursive bounds on $\norm{\Delta^kz^t}_E/k!$. The scalar coefficient comparison controls their growth geometrically; \cref{lem:high-order-logit} states the resulting centered-logit estimate.

For $t\ge2$, the weighted gaps satisfy $h_i^t=H_i(z^t,z^{t-1})$. Applying the same chain rule with the factorial derivative bound for $H_i$ gives
\[
\sup_{t\ge k+2}\norm{\Delta^kh_i^t}_2
\le C(R\sqrt n\,\eta)^k k!
\le C(R\sqrt n\,\eta k)^k.
\]
Thus the useful difference order is $m\asymp(\sqrt n\eta)^{-1}$: choosing its proportionality constant small enough makes the bound exponentially small in $m$. The single factorial is essential to this scale. The factor $\sqrt n$ comes from combining player blocks in the product norm, while probability normalization in $H_i$ avoids an action-count factor.

\paragraph{Comparing high and first differences directly.}
The choice of $h_i^t$ also addresses the moving-weight difficulty. Its squared Euclidean norm equals $\Var_{x_i^t}(\ell_i^{t-1})$, and its first difference controls prediction error up to the $O(\eta^2)$ weight correction of \cref{lem:weighted-gap}. The probabilities are part of the vector being differentiated, rather than part of a norm that changes with time.

After \cref{lem:boundary-extension} constructs a tapered, finitely supported sequence $y$ from $h_i^t$, finite-difference interpolation (\cref{lem:fourier-tools}) gives
\[
\norm{\Delta y}_{\ell_2}^2
\le\theta\norm y_{\ell_2}^2
 +\theta^{1-m}\norm{\Delta^my}_{\ell_2}^2.
\]
Here $\theta$ can remain fixed, independently of $m$. There is no need to compare adjacent orders under changing weights or to freeze those weights on time intervals. The cutoff contributes $O(m)$ boundary error, and the high-order bound controls the interpolation term by $O((T+m)e^{-cm})$ after choosing the constants as in \cref{sec:temporal-estimate}. This yields \cref{lem:explicit-temporal}.

\paragraph{Why the admissible step size is larger.}
The two changes work together: factorial temporal bounds permit $m\asymp(\sqrt n\eta)^{-1}$, and fixed-norm interpolation transfers that control directly to first differences. Taking $m$ of order $\log(T+2)$ permits the horizon-tuned constant step size $\eta_T\asymp[\sqrt n\log(T+2)]^{-1}$. The entropy term then gives the logarithmic regret bound in \cref{cor:logarithmic}. The improvement is therefore not a new RVU inequality or the use of Fourier analysis alone, but sharper high-order estimates in a representation that preserves the relevant probability weights. The chain rule and scalar comparison isolate a reusable question about discrete recurrences; applying them elsewhere would require corresponding uniform derivative bounds.

\end{document}